\documentclass[letterpaper]{article}
\usepackage[margin=1in]{geometry}

\usepackage{microtype}
\usepackage{hyperref}
\usepackage{amsmath,amsthm,amssymb}
\usepackage{algorithm}
\usepackage{algpseudocode}
\usepackage{graphicx}

\usepackage[utf8]{inputenc}

\usepackage{graphicx,subfigure}
\usepackage{caption}
\usepackage{enumerate}
\usepackage{xspace}
\usepackage{xcolor}
\usepackage{cleveref}
\usepackage{dsfont}
\usepackage{comment}
\usepackage{url}

\usepackage{thmtools}
\usepackage{thm-restate}

\usepackage{appendix}

\theoremstyle{plain}

\newtheorem{lemma}{Lemma}[section]

\newtheorem{corollary}[lemma]{Corollary}

\newtheorem{fact}[lemma]{Fact}

\theoremstyle{definition}
\newtheorem{definition}[lemma]{Definition}

\theoremstyle{remark}
\newtheorem{remark}[lemma]{Remark}

\hypersetup{colorlinks=true,allcolors=blue}

\title{Algorithms for the Generalized Poset Sorting Problem}

\author{
  Shaofeng H.-C. Jiang%
  \thanks{Email: \texttt{shaofeng.jiang@pku.edu.cn}
  }\\
  Peking University
  \and Wenqian Wang%
  \thanks{Email: \texttt{wangwenqian@sjtu.edu.cn}
  }\\
  Shanghai Jiao Tong University
  \and Yubo Zhang%
  \thanks{Email: \texttt{zhangyubo18@pku.edu.cn}
  }\\
  Peking University
  \and Yuhao Zhang%
  \thanks{Email: \texttt{zhang\_yuhao@sjtu.edu.cn}
  }\\
  Shanghai Jiao Tong University
}

\date{}

\newcommand{\algLE}{\textsc{Part-to-LE}}
\newcommand{\algPart}{\textsc{Partition}}

\newcommand{\Anc}{\mathrm{Anc}}
\newcommand{\CAnc}{\mathrm{CAnc}}
\newcommand{\IAnc}{\mathrm{IAnc}}
\newcommand{\sgn}{\mathrm{sgn}}
\newcommand{\nQuery}{\#\mathrm{queries}}
\newcommand{\Troot}{\mathrm{root}}
\newcommand{\Lset}[2]{#1_{\rightarrow #2}}
\newcommand{\Mset}[2]{#1_{\nleftrightarrow #2}}
\newcommand{\Rset}[2]{#1_{\leftarrow #2}}
\newcommand{\chL}{\mathrm{ch}_{\prec}}
\newcommand{\chM}{\mathrm{ch}_{\incom}}
\newcommand{\chR}{\mathrm{ch}_{\succ}}
\newcommand{\rnk}{\mathrm{rank}}
\newcommand{\perm}{\mathrm{perm}}
\newcommand{\pperm}{\mathrm{PathPerm}}
\newcommand{\lp}{\mathrm{lp}}

\newcommand{\Gbase}{\vec{G}_\mathrm{base}}
\newcommand{\Ebase}{\vec{E}_\mathrm{base}}
\newcommand{\uGbase}{G_\mathrm{base}}
\newcommand{\uEbase}{E_\mathrm{base}}
\newcommand{\pvt}{\mathrm{pivot}}
\newcommand{\gIn}{\mathrm{in}}
\newcommand{\gOut}{\mathrm{out}}
\newcommand{\gAdj}{\mathrm{adj}}
\newcommand{\algPartRand}{\textsc{PartitionRG}}

\newcommand{\opt}{\textup{OPT}\xspace}
\newcommand{\eopt}{\widetilde{\textup{OPT}}\xspace}
\newcommand{\alg}{\textup{ALG}\xspace}
\newcommand{\incom}{\nsim}
\newcommand{\p}{\Pr}
\newcommand{\sortchain}{\textsc{SortChain}\xspace}
\newcommand{\sortpartial}{\textsc{GPSC}\xspace}
\newcommand{\sortconstant}{\textsc{SortWeighted}\xspace}

\newcommand{\findj}{\textsc{FindThreshold}\xspace}

\newcommand{\findmin}{\textsc{FindMin}\xspace}

\newcommand{\findkmin}{\textsc{FindLarge}\xspace}

\newcommand{\finda}{\mathcal{E}^A\xspace}
\newcommand{\findb}{\mathcal{E}^B\xspace}

\newcommand{\vmin}{v^*\xspace} \newcommand{\bunion}{\tilde{B}\xspace}
\newcommand{\pivot}{p\xspace}

\newcommand{\W}{W}
\newcommand{\ER}{Erd\H{o}s-R\'{e}nyi\xspace}
\DeclareMathOperator{\poly}{poly}

\newcommand{\colnote}[3]{\textcolor{#1}{$\ll$\textsf{#2}$\gg$\marginpar{\tiny\bf #3}}}
\newcommand{\shaofeng}[1]{\colnote{purple}{#1--Shaofeng}{SJ}}

\begin{document}
    \maketitle

    \begin{abstract}
    We consider a generalized poset sorting problem (GPS),
    in which we are given a query graph $G = (V, E)$ and an \emph{unknown poset} $\mathcal{P}(V, \prec)$ that is defined on the same vertex set $V$,
    and the goal is to make as few queries as possible to edges in $G$
    in order to fully recover $\mathcal{P}$,
    where each query $(u, v)$ returns the relation between $u, v$, i.e., $u \prec v$, $v \prec u$ or $u \not \sim v$.
    This generalizes both the poset sorting problem [Faigle et al., SICOMP 88]
    and the generalized sorting problem [Huang et al., FOCS 11].

    We give algorithms with $\tilde{O}(n \poly(k))$ query complexity when $G$ is a complete bipartite graph or $G$ is stochastic under the \ER model, where $k$ is the \emph{width} of the poset, and these generalize [Daskalakis et al., SICOMP 11] which only studies complete graph $G$.
    Both results are based on a unified framework that reduces the poset sorting to partitioning the vertices with respect to a given pivot element, which may be of independent interest.
    Moreover, we also propose novel algorithms to implement this partition oracle. Notably, we suggest a randomized BFS with vertex skipping for the stochastic $G$, and it yields a nearly-tight bound even for the special case of generalized sorting (for stochastic $G$) which is comparable to the main result of a recent work [Kuszmaul et al., FOCS 21] but is conceptually different and simplified.

    Our study of GPS also leads to a new $\tilde{O}(n^{1 - 1 / (2W)})$ competitive ratio for the so-called weighted generalized sorting problem where $W$ is the number of distinct weights in the query graph.
    This problem was considered as an open question in [Charikar et al., JCSS 02],
    and our result makes important progress as it yields the first nontrivial sublinear ratio for general weighted query graphs (for any  bounded $W$).
    We obtain this via an $\tilde{O}(nk + n^{1.5})$ query complexity algorithm
    for the case where every edge in $G$ is guaranteed to be comparable in the poset, which generalizes a $\tilde{O}(n^{1.5})$ bound for generalized sorting [Huang et al., FOCS 11].
\end{abstract}
     \section{Introduction}
\label{sec:intro}

We consider a generalized poset sorting problem and obtain various new algorithmic results.
In the \emph{generalized poset sorting} problem (GPS), we are given an undirected \emph{query graph} $G = (V, E)$ and an unknown \emph{poset} $\mathcal{P} = (V, \prec)$.
The goal is to fully recover the poset $\mathcal{P}$, that is, to figure out the relation between all $x, y \in V$, through the smallest number of queries to the edges in $G$.
Here, when the algorithm makes a query $(u, v) \in E$, the relation of $u$ and $v$ in the poset, i.e., $u \prec v$, $v \prec u$ or $u \not \sim v$ (which stands for $u$ and $v$ are not comparable), is returned.

When the comparison graph $G$ is a complete graph, GPS 
reduces to a special case called the \emph{poset sorting} problem which was suggested by~\cite{DBLP:journals/siamcomp/FaigleT88}.
This poset sorting is a fundamental problem
since it captures the presence of incomparable elements in a partially ordered set which does not have a linear ordering.
For this problem, an algorithm with optimal $O(nk + n \log n) = \tilde{O}(nk)$  query complexity was given in~\cite{DBLP:journals/siamcomp/DaskalakisKMRV11} where $k$ is the \emph{width} of the poset.\footnote{Throughout, $\tilde{O}(f) := O(f \poly\log f)$.}
However, this $\tilde{O}(nk)$ bound heavily relies on the fact that $G$ is complete, and does not work for our general case where the query graph $G$ can have missing edges $(u, v)$ which forbid the query of the relation between $u$ and $v$ (we shall provide a more detailed technical discussion later).

In fact, the missing edges in the query graph already introduce significant challenges even when the poset $\mathcal{P}$ is a total order (where every two elements are comparable).
This special case (general graph $G$ and total order) is called
\emph{generalized sorting} whose study was initiated by~\cite{DBLP:conf/focs/HuangKK11}.
The state-of-the-art algorithm for this generalized sorting needs to use $\tilde{O}(\sqrt{mn})$ queries~\cite{DBLP:conf/focs/KuszmaulN21} for general graphs $G$, far from matching the $\Theta(n \log n)$ bound for the classic sorting. 
On the other hand, a parallel research theme
aims to explore whether $\tilde{O}(n)$ query complexity can be obtained for generalized sorting on special graph families.
Notably, such algorithms were obtained for complete bipartite graphs~\cite{DBLP:conf/soda/AlonBFKNO94,bradford1995matching,DBLP:journals/ipl/AlonBF96,DBLP:journals/siamdm/KomlosMS98} and \ER stochastic graphs~\cite{DBLP:conf/focs/HuangKK11,DBLP:conf/focs/KuszmaulN21}.

\paragraph{Our focus.}

Thus, a fundamental question is to figure out which families of query graphs $G$ admits algorithms with the optimal $\tilde{O}(nk)$ queries (to match that for the complete graphs~\cite{DBLP:journals/siamcomp/DaskalakisKMRV11}) for GPS, where $k$ is the width of the poset.
An ideal goal is to achieve this $\tilde{O}(nk)$ bound for general query graphs,
but as we mentioned, even for the total order case it is already difficult to improve over $\sqrt{mn}$.
Therefore, we instead focus on complete bipartite and \ER stochastic query graphs, which are fundamental cases and were very well studied in the special case of generalized sorting.
Moreover, we also study how GPS connects to other settings, especially its implications for variants of generalized sorting. 
This connection is plausible since a natural way for sorting is to build a partially sorted solution and then solve the remaining sub-problem, and this sub-problem may often be modeled as a poset sorting problem.

\paragraph{Technical challenges.}
However, designing algorithms for GPS turns out to be nontrivial and requires new approaches. 
Below, we briefly discuss why the existing techniques from tightly related problems, including poset sorting and generalized sorting problems, cannot be readily applied.
\begin{itemize}
    \item \textbf{Techniques from poset sorting.}
    The missing edges in $G$ can increase the query complexity of existing algorithms for poset sorting~\cite{DBLP:journals/siamcomp/FaigleT88,DBLP:journals/siamcomp/DaskalakisKMRV11}.
    In these algorithms, the overall framework is to incrementally add elements to the current sorting, and when an element $x$ is to be added, a binary search is used to figure out the relation between $x$ and every other element added so far.
    A crucial step to bound the query complexity is that there always exists a path cover of size $k$ (which is the width) in the induced subgraph of the added elements, and this ensures only $k$ binary search suffices.
    While this is true for complete graphs, the width $k$ can no longer upper bound the size of the path cover only by using edges in an induced subgraph of a general $G$.
    \item \textbf{Techniques from generalized sorting.}
    Incomparable edges $(u, v) \in E$ ($u \not \sim v$) reveal very little information about ordering, hence algorithms for generalized sorting should avoid querying these edges. 
    However, existing algorithms~\cite{DBLP:conf/focs/HuangKK11,DBLP:conf/sosa/LuRSZ21,DBLP:conf/focs/KuszmaulN21} for generalized sorting (designed for general query graph $G$) relies on a rough estimation of the relation between elements,
    and ``useful'' edges may be wrongly classified as incomparable edges.
    If this happens, then it is very difficult to detect the ``useful'' edge without querying a lot of incomparable edges, making existing algorithms less efficient.
\end{itemize}

\subsection{Our Results}
\label{sec:results}

We give efficient algorithms for GPS that make $\tilde{O}(n\poly(k))$ queries for \ER stochastic query graphs (\Cref{thm:er}) and complete bipartite query graphs (\Cref{thm:bipartite}), where $k$ is the width of the poset throughout (see \Cref{sec:models} for formal definitions of these query graph models).
These are the first results for GPS parameterized by the width of the poset $k$, and the query complexity bound is nearly-optimal.
We obtain our results via a unified framework and it may be of independent interest (will be discussed in \Cref{sec:tech_overview}).
These results are our main technical contributions.

\begin{restatable}[\ER Stochastic Graphs]{theorem}{MainThmRG}
    \label{thm:er}
    There exists an algorithm that solves GPS
    on \ER stochastic query graphs and width-$k$ posets
    using $\tilde{O}(nk^2)$ queries with high probability.
    This holds regardless of the probability parameter $0 < p \leq 1$ in \ER $G(n, p)$.
\end{restatable}

Although our bound for \ER stochastic query graph does not depend on $p$,
it still relies on the structural property of the \ER graph where edges are i.i.d. generated.
This case of \ER query graph has been well studied in (total order)
generalized sorting (i.e., $k = 1$), where~\cite{DBLP:conf/focs/HuangKK11} and~\cite{DBLP:conf/focs/KuszmaulN21} are milestones.
Compared with~\cite{DBLP:conf/focs/HuangKK11},
our result is significantly better than their $\min\{ np^{-2}, n^{1.5}\sqrt{p} \}$,
especially that our algorithm is $\tilde{O}(n)$ regardless of $p$ and theirs can obtain near-linear query complexity only for a very limited range of $p$.
On the other hand, compared to the more recent work~\cite{DBLP:conf/focs/KuszmaulN21} whose bound is $O(n \log(np))$,
our result is worse by a $\poly\log n$ factor.
However, our slightly worse bound excels in that it is conceptually simpler and technically different, plus it generalizes to poset sorting. See \Cref{sec:tech_overview} for a more detailed discussion.
Finally, we remark that a unique feature of~\cite{DBLP:conf/focs/KuszmaulN21} is that when $p$ is very small, say $p = 1 / n$, then the complexity of sorting can even be better than $O(n \log n)$ which is the well-known lower bound for classic sorting.
We leave it as an open question to figure out if one can achieve a similar bound for GPS with \ER query graphs.

\begin{restatable}[Complete Bipartite Graphs]{theorem}{thmbipartite}
\label{thm:bipartite}
    There exists an algorithm that solves GPS 
    on complete bipartite query graphs and width-$k$ posets  
    using $\tilde{O}(nk)$ queries with high probability.
\end{restatable}
This result for complete bipartite graphs is tight up to $\poly(\log n)$ factors, since 
the $\Omega(nk)$ lower bound for complete query graph in~\cite{DBLP:journals/siamcomp/DaskalakisKMRV11} still holds in the complete bipartite case (see \Cref{remark:bipartite}). Our result is also a generalization of a series work of nuts-and-bolts problems~~\cite{DBLP:conf/soda/AlonBFKNO94,bradford1995matching,DBLP:journals/ipl/AlonBF96,DBLP:journals/siamdm/KomlosMS98}, and our bound for $k=1$ nearly matches the state-of-the-art bound for this problem (up to $\poly\log$ factors, compared to the best $O(n\log n)$ result in (total order) generalized sorting problem on complete bipartite graphs by \cite{DBLP:journals/siamdm/KomlosMS98}). 
Note that there is a difference between the nuts-and-bolts problem and the generalized sorting problem on complete bipartite graphs, where the nuts-and-bolts problem has an additional assumption that every node is assigned an edge with the query result ``equal". With the help of the ``equal" edge, a natural randomized quicksort-like algorithm achieves the query complexity $O(n \log n)$, by an $O(n)$ partition algorithm to partition nodes based on a randomly selected pivot. However, if the ``equal" edge is not provided, it is already a non-trivial task to design a partition algorithm. This is also noted by \cite{DBLP:journals/siamdm/KomlosMS98}, and they resolve the missing ``equal" edges and design an $O(n \log n)$ sorting algorithm with some other indirect methods. However,  this indirect method does not yield a partition algorithm, and we find it hard to generalize these indirect methods to the setting of poset.
In our result, we devised a partition algorithm for complete bipartite graphs without ``equal" edges, and this type of partiton algorithm was unknown even for the total order setting.
This partition step turns out to be useful and naturally generalizable to posets.

\paragraph{Weighted generalized sorting.}
Apart from the significance in its own right, another important implication of GPS is that it can be used as an intermediate step for other (total order) sorting problems.
We showcase this idea by presenting new results for the weighted generalized sorting problem.

In the weighted (total order) generalized sorting problem,
the query graph is weighted (with weight function $w : E \to \mathbb{R}_{\geq 0}$),
and each query $(u, v) \in E$ incurs a weighted cost $w(u, v)$ instead of a unit cost.
Since the objective is weighted, we measure the performance of the algorithm using the \emph{competitive ratio}, defined as the total cost incurred by the algorithm divided by $\sum_{i}w(v_i, v_{i+1})$, where $v_1 \prec \ldots \prec v_n$ is the total order.
We obtain the following result for weighted generalized sorting.

\begin{restatable}{theorem}{thmweighted}
\label{thm:weighted}
There exists an algorithm
that solves the weighted (total order) generalized sorting problem with competitive ratio $\displaystyle\tilde{O}(n^{1-1/(2W)})$, where $W$ is the number of distinct weights in the graph.
\end{restatable}

Indeed, obtaining nontrivial bounds for this weighted generalized sorting has been suggested as an \emph{open question} by~\cite{DBLP:journals/jcss/CharikarFGKRS02},
and it received significant attention in various subsequent works~\cite{DBLP:conf/focs/GuptaK01,DBLP:conf/latin/AngelovKM08,DBLP:journals/corr/BanerjeeR15a,DBLP:journals/corr/abs-2211-04601}.
However, these existing works mostly focus on understanding certain special cases of weights, such as bounded number of distinct weights~\cite{DBLP:journals/corr/BanerjeeR15a,DBLP:journals/corr/abs-2211-04601} or structured/random weights~\cite{DBLP:conf/focs/GuptaK01,DBLP:conf/approx/GuptaK05,DBLP:conf/latin/AngelovKM08}.
For the general case, we are only aware of an $O(n)$ ratio~\cite{DBLP:conf/focs/GuptaK01},
which is trivial in the unweighted case since one can query all edges,
but is already nontrivial in the weighted setting.

Our result makes progress on the weighted generalized sorting problem, 
and our bound implies a strictly sublinear ratio when the number of distinct weights is bounded (and the weights can take any non-negative values).
This improves over the known $O(n)$ ratio in~\cite{DBLP:conf/focs/GuptaK01},
and our ratio also matches ratios for several notable special cases.
When $W = 1$ which reduces to (total order) generalized sorting, 
our bound matches the $\tilde{O}(n^{1.5})$  query complexity
in~\cite{DBLP:conf/focs/HuangKK11} which is the state-of-the-art for dense graphs (for sparse graphs a $\sqrt{mn}$ bound was obtained in~\cite{DBLP:conf/focs/KuszmaulN21}, where $m$ is the number of edges in the query graph).
Moreover, we give in \Cref{cor:weighted_fast_inc} an improved analysis of our algorithm for the case when the weights are well-separated,
and this result matches an $\tilde{O}(n^{3 /4})$
ratio, obtained in a recent work~\cite{DBLP:journals/corr/abs-2211-04601}, for the case when the weights are picked from $\{0, 1, n, \infty\}$ (where $\infty$ weight can be interpreted as ``missing edge'' in the query graph in our case).

\begin{remark}
In several previous works~\cite{DBLP:conf/focs/GuptaK01,DBLP:conf/focs/HuangKK11,DBLP:conf/focs/KuszmaulN21}
it has been mentioned that the ratio for finding a maximum element in a weighted query graph is $\Omega(n)$,
and this seems to suggest that $\Omega(n)$ is also a lower bound for sorting (since sorting implies finding the maximum).
However, this is not true, for two reasons.
One is that the ratio in the maximum-finding problem is defined with respect to the cost of a minimum-weight certificate for the maximum,
which can be much smaller than the $\sum_{i}w(v_i, v_{i + 1})$ cost for identifying total order in sorting.
Secondly, the hard instance in the lower bound of maximum-finding~\cite{DBLP:conf/focs/GuptaK01} only uses three types of weights,
and by our upper bound, this type of instance cannot be hard for sorting.
A similar discussion of this gap was also made in~\cite{DBLP:journals/corr/abs-2211-04601}.
In fact, a further implication of our result is that,
in order to  prove $\Omega(n)$ lower bound for weighted generalized sorting,
one must use at least $\Omega(\log n)$ distinct weights in the hard instance. 
\end{remark}

\paragraph{Auxiliary problem: GPS with comparable edges (GPSC).}
As we mentioned, GPS is used as an important intermediate step for obtaining \Cref{thm:weighted}.
In particular, we consider a special case of GPS whose query graph consists of \emph{comparable} edges only, i.e., $(u, v) \in E$ only if $u \prec v$ or $v \prec u$, and we call this special case the GPS with comparable edges (GPSC).
Due to the fact that GPSC is in between GPS and (total order) generalized sorting, and that it may be useful for other sorting problems,
the result of this problem, stated below, may be of independent interest.

\begin{restatable}[GPSC]{theorem}{thmgpsc}
\label{thm:gpsc}
    There exists an algorithm that solves GPSC 
    on general query graphs and width-$k$ posets  
    using $\tilde{O}(nk + n^{1.5})$ queries with high probability.
\end{restatable}

As mentioned, \Cref{thm:gpsc} is a crucial subroutine for \Cref{thm:weighted},
but to obtain a sublinear ratio for weighted generalized sorting (provided that the number of distinct weights is bounded),
any $n^{2 - \epsilon}$ query bound for GPSC suffices,
although it may lead to a worse constant in the exponent of $n$ in the ratio (i.e., worse than $O(n^{1 - 1 / (2W)})$ but still $o(n)$).
We give a detailed overview on how this can be used to obtain \Cref{thm:weighted} in \Cref{sec:tech_overview}.

\subsection{Technical Overview}
\label{sec:tech_overview}

We give a highlight of technical challenges and our technical contributions, followed by a more detailed technical overview.
\begin{itemize}
    \item A genaeral framework for GPS. Previous algorithms for generalized sorting~\cite{DBLP:conf/focs/HuangKK11, DBLP:conf/sosa/LuRSZ21, DBLP:conf/focs/KuszmaulN21} 
    use an incremental method to iteratively discover the (nearly-)minimum element,
    but this does not work directly in GPS due to incomparable edges and non-unique minimal elements.
    We develop a general framework for GPS, which reduces GPS to finding a linear extension, 
    and we further show this linear extension can be found by a quicksort-like algorithm proposed by \cite{DBLP:conf/focs/HuangKK11,DBLP:journals/siamcomp/DaskalakisKMRV11},
    but we would need a new analysis to save a factor of $k$ in the query complexity.
    Specifically, our new analysis requires a stronger partition algorithm with a refined query complexity bound that depends on the width of the poset in the subproblem.
Since it does not introduce any additional $\poly(n)$ or $\poly(k)$ factors,
our framework is capable of obtaining (nearly) tight bounds when combined with carefully designed downstream partition algorithms, which may be of independent interest. 
    \item Novel partition algorithms for \ER graphs based on stochastic BFS. Our partition algorithm for \ER graphs is based on a stochastic BFS,
    where the key idea is to skip a vertex from the BFS queue if that vertex has been visited by sufficiently many other vertices.
    This still guarantees the correctness with high probability due to the property of \ER graph. To make sure we trim most vertices in a few iterations, we also run the BFS in a random order of vertices.
    Previously, algorithms for \ER graphs were only known for generalized sorting (without considering a poset),
    and the techniques are not readily applicable.
    In particular, the framework of~\cite{DBLP:conf/focs/HuangKK11} requires an algorithm with a subquadratic query for general query graph, which is not available in GPS.
    Another recent work \cite{DBLP:conf/focs/KuszmaulN21} uses a very different approach,
    but the efficiency of one of its subroutines relies on the uniqueness of the minimal element. Hence it is highly nontrivial to generalize to the poset setting while achieving subquadratic complexity.
\item Weighted generalized sorting via new algorithms for GPSC.
    To achieve the $\tilde{O}(n^{1-\frac{1}{2W}})$ competitive ratio,
    we partition the edge set into cheap and expensive edges according to a threshold,
    and the two cases are balanced and solved by one of the following two algorithms:
    a) a new sorting algorithm that may receive a partially sorted graph (i.e. a partial order) as extra input and the competitive ratio depends on the width of the input partial order; and b) a new GPSC algorithm as in \Cref{thm:gpsc}.
    In our new GPSC algorithm, we employ the framework of sorting with predictions~\cite{DBLP:conf/focs/HuangKK11,DBLP:conf/sosa/LuRSZ21, DBLP:conf/focs/KuszmaulN21} (which was proposed for generalized sorting),
    where we construct a prediction graph that ``guesses'' the direction of the edges,
    and make decisions and refine the prediction in an iterative manner.
    To ensure the ratio is linear in $k$, we devise a stronger predictor that has an ``everywhere'' guarantee for each vertex, as opposed to having a collective bound on the total number of wrongly predicted edges.
    
\end{itemize}

\paragraph{A general framework for generalized poset sorting.}
As mentioned, we obtain algorithms for GPS via a new unified framework.
In this framework, we first reduce the GPS to finding a \emph{linear extension} (\Cref{thm:linearexttoposet}).
A linear extension for the poset $\mathcal{P} = (V, \prec)$ is a total order such that $\forall x, y \in V$, if $x \prec y$ then $x$ appears before $y$ in the total order.
Finding a linear extension is an interesting problem in its own right,
and it has also been studied in~\cite{DBLP:conf/focs/HuangKK11,DBLP:journals/siamcomp/DaskalakisKMRV11}.
However, previous studies did not establish the connection between GPS and linear extension, which we do in our framework.

To find the linear extension, we employ a quicksort-like algorithm
to randomly select a pivot element $v \in V$ and partition the elements into three parts, elements smaller than $v$, elements incomparable with $v$ and elements larger than $v$.
Given this partition, one can compute the linear extension of these three parts recursively and combine them in the order of smaller-incomparable-larger to obtain the linear extension of $\mathcal{P}$.

This quicksort-like algorithm was also used in~\cite{DBLP:journals/siamcomp/DaskalakisKMRV11} to find a linear extension for complete query graphs.
While their analysis may be adapted to the general query graph case,
it only leads to sub-optimal bounds with respect to $k$.
We give new analysis to this quicksort-like algorithm,
and we are able to obtain an improved dependence in $k$  (\Cref{lem:part2LE-main,lem:part2LE-main2})
provided that the partition algorithm additionally satisfies certain properties.
These properties turn out to be natural, and we manage to design partition algorithms satisfying these properties for both \ER and complete bipartite query graphs.

Now we explain our new steps in the analysis to the quicksort-like algorithm.
In~\cite{DBLP:journals/siamcomp/DaskalakisKMRV11},
it is observed that the depth of the recursion tree is $O(k + \log n)$.
This is good enough for complete query graphs, since the partition step can be done in $O(n)$, and this, combined with the depth of the recursion tree, translates to an $O(nk+n\log n)$ bound. 
However, when $G$ is not a complete graph, the partition problem often requires $\Omega(n)$ queries, say $O(nk^c)$ queries, then the analysis in~\cite{DBLP:journals/siamcomp/DaskalakisKMRV11} leads to an $\tilde{O}(nk^{c+1})$ bound, which introduces an additional $k$ factor.
In order to avoid this additional $k$ factor, we require partition algorithms to use $O(n k_v^c)$ queries that depend on $k_v$, which denotes the width of elements comparable with pivot vertex $v$.
A crucial observation is that, if $k_v$ is small, then the partition algorithm uses few queries, and if $k_v$ is large, then the next pivots $v'$ (in the incomparable part) is likely to have a small $k_{v'}$. 

\begin{comment}
If one directly uses the approach in~\cite{DBLP:journals/siamcomp/DaskalakisKMRV11},
one would obtain a bound like $k \cdot f$ where $f$ is the query complexity of the partition algorithm, essentially introducing an addition $k$ factor.
In our bounds (\Cref{lem:part2LE-main,lem:part2LE-main2}), we manage to avoid suffering this additional $k$ factor.
Roughly speaking, we show if the partition algorithm runs in $\tilde{O}(n k^c)$ (as long as $c \geq 1$) then the full algorithm for linear extension runs in nearly the same order with respect to both $n$ and $k$ (up to $\poly \log$ factors).
\shaofeng{We can consider expanding this a bit -- maybe about how we obtain this improved bound?}
\end{comment}

\paragraph{Partition algorithms.}

For the partition step,
if it were the complete graph case, we could directly query the relations between the pivot and every other element using $n-1$ queries.
However, this simple but efficient bound is no longer easily obtainable when the query graph is not complete.
Nonetheless, we introduce novel ideas for this partition step, and we manage to obtain algorithms that use $\tilde{O}(nk^2)$ queries for \ER query graphs and $\tilde{O}(nk)$ queries for complete bipartite query graphs.

\paragraph{Partition algorithms: \ER graphs.}

It is helpful to interpret the problem as a graph problem.
We define a directed graph $\vec{G}$ from $G$,
by defining the direction of every edge $(u, v) \in E$ according to the relation between $u,v$, i.e. the direction is $u\to v$ if and only if $u\prec v$. 
Then for every vertex $u$, $u$ is smaller than the pivot vertex if and only if there exists a path from $u$ to pivot in $\vec{G}$.
Hence, the partition problem reduces to finding all vertices that can be reached from the pivot vertex. 
This graph problem may be solved using BFS, but a vanilla BFS needs to query all edges, which is too costly. To resolve this issue,
we design a variant of BFS that can make use of the structure of \ER graphs, called Skip-BFS.

We start by giving the overall intuition by assuming we are given a chain decomposition of the poset (which is of size $k$, guaranteed by Dilworth's Theorem).
An important property of \ER $G(n,p)$ is that, if we select $O(p^{-1} \log n)$ arbitrary vertices, then every vertex is adjacent to at least one selected vertex\footnote{This does not always happen and only with high probability, but in the following discussions we ignore this and talk about the typical behavior.}.
Hence, if we select the $O(p^{-1} \log n)$-largest vertices from each chain in the chain decomposition of the poset
then every vertex has outgoing edges to at least one selected vertex.
Exploring (the neighbor of) these selected vertices only takes $\tilde{O}(kp^{-1} \cdot np) = \tilde{O}(nk)$ queries, and this finishes the partition. 

However, the chain decomposition is not known to our algorithm a priori.
Thus, we need a method to gauge whether a vertex is worth exploring, i.e., it is sufficiently large in its chain.
To this end,
Skip-BFS maintains a counter $c[v]$ for each vertex $v$, which is initialized as some parameter $R = \Theta(\log n)$.
Then, if some vertex $v$ becomes the current vertex for which we start to explore its neighbor, we decrease the counter $c[u]$ by $1$ for every $v$'s neighbor $u$ such that $u$ is smaller than $v$.
When the counter of some vertex $u$ is decreased to $0$, we skip this point $u$ by removing it from the BFS queue.
Such $u$ can be safely skipped since Skip-BFS has already explored $R \geq \Omega(\log n)$ vertices that are larger than $u$,
and these vertices are likely to cover all incoming vertices of $u$.
To see this, since $u$'s counter is decreased $R \geq \Omega(\log n)$ times,
we already visited $O(p^{-1} \log n)$ vertices that are larger than $u$, and that each such vertex connects to $p$ fraction of vertices smaller than $u$. 
Hence, these already-visited $O(p^{-1} \log n)$ vertices connects to/cover all vertices that are smaller than $u$.
Finally, to guarantee the efficiency of this process, we need to use a random permutation of vertices when we do BFS, in order to trim most $u$'s in only a few steps.
This eventually leads to an $\tilde{O}(nk^2)$ time partition algorithm for \ER query graphs.

Compared with the approach in~\cite{DBLP:conf/focs/KuszmaulN21} who gave 
an algorithm for the total order case that uses $O(n\log (np))$ queries which is tight,
our bound is comparable, but our approach is conceptually different and simplified.
In fact, it is unclear if their approach can be efficiently generalized to the poset case.
In their algorithm, they repeatedly find the minimum vertex of the current graph and remove it.
To find the minimum vertex, they identify a set of candidate vertices and then trim the wrong ones by testing if there is an incoming edge, which requires querying the edges between the candidates and other vertices.
However, in GPS, there are multiple minimal vertices,
and it is nontrivial to bound the number of candidates and the adjacent edges to query since one cannot stop before one is certain that the surviving candidates are minimal.
Hence, it is nontrivial to generalize their approach to GPS using even subquadratic queries.

\paragraph{Partition algorithm: complete bipartite graphs.}

Suppose the two parts of the bipartite graph are $A$ and $B$, and suppose the pivot is $b \in B$.
Let $A_{\succ b}$ and $B_{\succ b}$ be the elements in $A$ and $B$ that are greater than $b$, respectively.
We focus the discussion on finding the elements that are larger than $b$, i.e., $A_{\succ b} \cup B_{\succ b}$.
Since the graph is complete bipartite, it is easy to obtain $A_{\succ b}$,
but it is nontrivial to obtain $B_{\succ b}$ since one cannot directly compare any other point in $B$ with the pivot $b$.
A natural idea to deal with this is to find the minimal elements in $A_{\succ b}$ so that one can figure out $B_{\succ b}$ from these elements.
However, finding the minimal element is technically nontrivial
even in the total order setting (whose minimal is unique),
let alone there may be multiple minimal elements in posets.
Indeed, the existing algorithm for the total order setting does not seem to make progress on this simple and fundamental task, and they solved the problem via other indirect methods~\cite{DBLP:conf/soda/AlonBFKNO94,bradford1995matching,DBLP:journals/ipl/AlonBF96,DBLP:journals/siamdm/KomlosMS98}.
We provide a completely new algorithm to find the minimal elements for poset in bipartite complete graphs, which is a technical contribution to the study of sorting and selection for bipartite graphs.

We first devise a \textsc{FindMin} procedure that finds a ``local'' minimal element in $A_{\succ b}$ (the ``local'' is due to the fact that we may make iterative calls and only run the procedure on an induced subgraph).
Then, we apply this \textsc{FindMin} iteratively to both find the minimal of $A_{\succ b}$ and construct $B_{\succ b}$.
In particular, every time we run \textsc{FindMin} to obtain a vertex $a$,
we try to find $B_{\succ a}$ and expand the currently found $B_{\succ b}$,
and remove $a$ from $A$ to continue.
Each iteration takes $\tilde{O}(n)$ queries.
Then using the property and the randomness of \textsc{FindMin}, the new element $a'$ that we find must be smaller, and this also shrinks the distance from $a$ to the pivot by a constant factor with good probability.
Finally, if one takes one chain in a chain decomposition, this entire process would typically run on a vertex from this chain for $O(\log n)$ iterations.
Summing over $k$ chains, the total query time is $\tilde{O}(nk)$ in the typical case.

\paragraph{GPSC.}

Recall that there are no incomparable edges 
in the query graph (which means every edge $(u, v)\in E$ satisfies either $u \prec v$ or $v \prec u$) of the GPSC problem. 
This conceptually simplifies the problem, since this avoids the issue of gaining essentially no information from querying an incomparable edge.
Technically, this allows us to apply techniques/frameworks developed for generalized sorting problem, which crucially relies on the information gain from querying an edge.
Specifically, we use an idea proposed in~\cite{DBLP:conf/focs/HuangKK11} and further developed in~\cite{DBLP:conf/focs/KuszmaulN21},
where one first constructs a prediction graph which ``guesses'' the direction/relation of all edges in the query graph.
Then, an incremental algorithm that iteratively adds a currently ``minimal'' element, i.e., an element $u$ with small number of ``incoming edges'' (which are the edges $(v, u)$ such that $v \prec u$) in the prediction,
is employed to generate the sorting.

To apply this framework to GPSC, especially to achieve a linear dependence in $k$,
we cannot use~\cite{DBLP:conf/focs/HuangKK11,DBLP:conf/focs/KuszmaulN21} in a black-box way,
since we need a stronger predictor
such that it has a bounded number of wrongly predicted edges \emph{everywhere}:
$\forall v$, there are $\tilde{O}(\sqrt{n})$ wrongly predicted edges among all adjacent edges to $v$.
This is stronger than the previously designed predictors~\cite{DBLP:conf/focs/HuangKK11,DBLP:conf/focs/KuszmaulN21}, since~\cite{DBLP:conf/focs/HuangKK11} only guarantees an $\tilde{O}(\sqrt{n})$ absolute error for the in-degree of every vertex (instead of the edge predictions),
and~\cite{DBLP:conf/focs/KuszmaulN21} only guarantees an overall number of wrong edges (instead of our ``everywhere'' guarantee).
We provide such a stronger predictor in \Cref{lem:wrongedge}.

Then, with this predictor, we iteratively maintain a current set $A$ of sorted vertices,
and we show it is possible to identify a key vertex $v \in A$,
whose incoming vertices (with respect to the prediction) can be partitioned into $X_v \subseteq A$ and $Y_v \cap A = \emptyset$, such that
the poset induced by $X_v$ still has width $k$, and that $|Y_v| = O(\sqrt{n})$ (by the stronger guarantee of the predictor).
This, together with the fact that $X_v$ can be decomposed into $k$ chains, implies that one can verify/discover all incoming edges (from the predictor) to $v$ using $O(k \log n + \sqrt{n})$ queries.
In total, this entire iterative process of updating $A$ happens $O(n)$ times, which leads to our final query bound $\tilde{O}(nk + n^{1.5})$.
Notice that our algorithm completely relies on the information in the prediction,
but this still suffices for the correctness by an argument similar to~\cite{DBLP:conf/sosa/LuRSZ21}.

\paragraph{Weighted generalized sorting.}
We start with designing an $O(k \poly\log n)$-competitive algorithm $\mathcal{A}$,
whose input consists of a chain decomposition (of the total order) of size $k$ in addition to the weighted query graph (and the underlying total order), for weighted generalized sorting.
Notice that one can always feed a trivial chain decomposition of size $n$ to  $\mathcal{A}$ and obtain $\tilde{O}(n)$ competitive ratio, which is already nontrivial as we mention in \Cref{sec:results}.
Although the algorithm by~\cite{DBLP:conf/focs/GuptaK01} can also achieve an $O(n)$ ratio for weighted generalized sorting,
it only works for the case when all chains are single nodes (i.e., $k=n$), hence it is not useful for obtaining a sublinear ratio. 

Next, we employ a threshold algorithm to ``combine'' $\mathcal{A}$ with \Cref{thm:gpsc} to obtain a sublinear ratio when the number of distinct weights is bounded.
Suppose the weights are $w_1 <  \ldots < w_\ell$.
We use a threshold parameter $1 \leq \tau \leq \ell$, and define  $G_\tau$ as the subgraph of the query graph with edge weights \emph{at most} $w_\tau$.
We also consider the poset $\mathcal{P}_\tau$ induced by $G_\tau$, and let $k_{\tau}$ denote its width.
We start with running \Cref{thm:gpsc} on $(G_\tau, \mathcal{P}_\tau)$ and ignoring the weight, which takes $\tilde{O}(n k_\tau)$ queries (assuming $k_\tau \geq \sqrt{n}$ in this discussion),
and it generates a chain decomposition of $\mathcal{P}_{\tau}$.
Notice that this chain decomposition of $\mathcal{P}_{\tau}$ is also a chain decomposition of $\mathcal{P}$ since they are supported on the same element set.
Then, we feed this chain decomposition to $\mathcal{A}$, and use the output of $A$ as the result.
The entire algorithm achieves an $\tilde{O}(\frac{n k_\tau \cdot w_\tau}{\opt} + k_\tau)$ ratio, and we can further show this ratio is at most $O(\frac{n w_\tau}{w_{\tau + 1}})$ (assuming that $k_\tau$ is not the dominating factor), which depends on the ``gap'' between two adjacent weights.
The final result can be achieved by fine-tuning of $\tau$ to minimize this ratio: if all weights are of a small gap, then one can view it as the unweighted case and run \Cref{thm:gpsc} directly, and otherwise, we have a significant gap which still allows a sublinear ratio.

\subsection{Related Work}

Parameterization other than the width of the poset which we use was also considered in the literature,
and they are generally not comparable to our results.
In~\cite{DBLP:conf/swat/BanerjeeR16,DBLP:conf/caldam/BiswasJ017},
the GPS is parameterized by the number of missing edges $q = \binom{n}{2} - m$ (where $m$ is the number of edges in the query graph) while there is no restriction on the poset,
and nearly-tight bounds were obtained with respect to $q$.
In a recent work~\cite{DBLP:journals/corr/abs-2205-15912}, the query graph can be general but the poset is assumed to be a tree and is parameterized by the maximum degree $d$, and they also obtained nearly tight query complexity bounds.

In addition to generalized sorting problems, other related problems were also considered.
Examples include noisy sorting/selection~\cite{DBLP:journals/siamcomp/FeigeRPU94,DBLP:conf/soda/BravermanM08,noisy_sorting} and generalized/weighted selection~\cite{DBLP:conf/focs/GuptaK01,DBLP:conf/soda/KannanK03,DBLP:conf/latin/AngelovKM08,DBLP:journals/siamcomp/DaskalakisKMRV11}.     \section{Preliminaries}\label{sec:preliminary}

Throughout, we use $\mathcal{P}=(V,\prec)$ to denote a poset, and we let $k_{\mathcal{P}}$ denote the width of $\mathcal{P}$. 
By Dilworth's Theorem, a poset of width $k$ can be decomposed to $k$ chains, say $\mathcal{C}=\{C_1,C_2,\dots, C_k\}$, where the elements are comparable to each other on each chain. 
Suppose $\mathcal{P} = (V,\prec)$ is the underlying poset of GPS problem, we directly use $k$ to denote $k_\mathcal{P}$. 
For every $X\subseteq V$, let $k_X$ denotes the width of poset $(X,\prec)$. 
For every $X\subseteq V, v\in X$, let $X_{\prec v}, X_{\incom v}, X_{\succ v}$ denote the elements that are smaller than $v$, incomparable with $v$ and larger than $v$ respectively. Recall that $x\incom y$ denotes ``$x$ is incomparable with $y$''. For some set $X$, denote the set of permutations of $X$ by $\perm(X)$. For every set $X$ with a total order $(X,\prec)$ and every $x\in X$, define $\rnk_X(x) = |\{ y\in X \mid y \prec x \}| + 1$ as the rank of $x$. Similarly, let $\rnk_{\mathbf{p}}(x)$ be the rank of $x$ in permutation $\mathbf{p}$. 

For a graph $G = (V, E)$ and a vertex subset $S \subseteq V$, denote $G[S]$ as the induced subgraph of $G$ on $S$, whose vertex set is $S$ and the edge set is $\{ (u, v) \in E : u, v \in S \}$.
Given a directed acyclic graph (DAG) $\vec{G}(V,\vec{E})$,
let $\mathcal{P}(\vec{G})$, which stands for induced poset of $\vec{G}$, be a poset $\mathcal{P}'(V, \prec)$,
such that $\forall u, v \in V$, $u \prec v$ if and only if there is a directed path from $u$ to $v$ in $\vec{G}$.
This implies that $u \not \sim v$ if and only if $u$ cannot reach $v$ and $v$ cannot reach $u$ in $\vec{G}$. 
By Dilworth's Theorem, a DAG can be covered by $k$ paths, and these paths form a path cover of $\vec{G}$, where $k$ is the width of $\mathcal{P}(\vec{G})$, i.e., every vertex is contained in at least one path (and may be contained in multiple paths).

\subsection{Models}
\label{sec:models}
\paragraph{Generalized poset sorting (GPS).}
Formally, in the GPS problem, we are given an $n$-element underlying (unknown) poset $\mathcal{P}=(V,\prec)$ and a graph $G = (V, E)$.
An oracle receives queries of the form $(u, v) \in E$, and returns the relation of $u, v$ in $\mathcal{P}$.
The goal of GPS is to use the minimum number of queries to fully recover $\mathcal{P}$, i.e., $\forall u, v \in V$, correctly determine the relation of $u, v$ in $\mathcal{P}$.

\paragraph{Model of query graphs in GPS.}
To make sure the problem is well-defined, e.g., $G$ has sufficient edges to recover $\mathcal{P}$,
we need to add some further constraints on the query graph $G$.
Specifically, we enforce the following:
let $\vec{E} := \{ (u, v) \in E : u \prec v \}$ and $\vec{G} = (V, \vec{E})$ (noting that $\vec{G}$ is defined with respect to both $G$ and $\mathcal{P}$)
    then
    \begin{equation}
        \label{eqn:p_induce}
        \mathcal{P} = \mathcal{P}(\vec{G}).
    \end{equation}
    This is well-defined if $G$ is deterministic (for instance $G$ is a complete bipartite graph),
    but for stochastic case enforcing this directly may cause randomness issues.
    Hence, we discuss how we define the \ER stochastic query graph in more detail in the following.

\paragraph{Model of query graphs in GPS: \ER stochastic case.}
    Let $G(n, p)$ denote the \ER random graph with $n$ vertices and probability parameter $0 \leq p \leq 1$.
    Specifically, this $G(n, p)$ is generated by independently adding an undirected edge $(u, v)$ with probability $p$ for every vertex pair $u \neq v$.
Clearly, this random graph is unlikely to be able to uniquely identify $\mathcal{P}$.
Hence, we still wish to enforce the property stated in \eqref{eqn:p_induce}.
Specifically,
we need to add to the \ER graph a \emph{minimal} DAG $\uGbase$ which is a ``base graph''.
Here, we say a DAG $\vec{G}$ is minimal if there is no redundant edge in $\vec{G}$, where we call an edge $u\rightarrow v$ redundant if we have $\exists x \notin \{u, v\}$, $u\rightarrow x$ and $x \rightarrow v$.
Formally, we have the following definition, and it indeed satisfies \eqref{eqn:p_induce} (stated in \Cref{fact:p_induce_er}).
\begin{definition}
    \label{def:er}
    Fix some minimal DAG $\Gbase$ such that $\mathcal{P} = \mathcal{P}(\Gbase)$,
    denoting its underlying undirected graph as $\uGbase$,
    the \ER stochastic query graph $G$ is a union of $\uGbase$ and $G(n,p)$.
\end{definition}
\begin{fact}
    \label{fact:p_induce_er}
    The random query graph $G = (V, E)$ defined in \Cref{def:er} satisfies \eqref{eqn:p_induce} with probability $1$, namely, $\Pr[ \mathcal{P} = \mathcal{P}(\vec{G}) ] = 1$ where $\vec{G} := \{ (u, v) \in E \mid u \prec v \}$.
\end{fact}
Indeed, this definition can be viewed as a generalization of the stochastic setting in the (total order) generalized sorting~\cite{DBLP:conf/focs/HuangKK11,DBLP:conf/focs/KuszmaulN21}, 
here they use a directed Hamiltonian Path as a base graph (which is the only minimal choice in the total order setting).

\begin{comment}
\paragraph{Linear extension.}
As in \Cref{thm:linearexttoposet}, it suffices to construct a \emph{linear extension} in order to recover $\mathcal{P}$.
A linear extension of the post is an ordered sequence $(v_1,v_2,\dots,v_n)$ where we cannot have $u_i \prec u_j$ if $i > j$. Remark that it is also a topological order of $\vec{G}$.  

\paragraph{Partition with respect to a pivot.}
As in \Cref{thm:PartToLE}, the problem of constructing linear extension further reduces to 
generating with respect to a given pivot point $\pivot \in V$ a $3$-partition of $V$
that corresponds to the elements that are smaller, larger and incomparable to $V$, respectively.
Formally speaking, we define 
\begin{align*}
V_{\succ \pivot} = \left\{ v \in V \mid \ v \succ  \pivot\right\}, \quad 
V_{\prec \pivot} = \left\{ v \in V \mid \ v \prec  \pivot\right\}, \quad 
V_{\nsim\pivot} = \left\{ v \in V \mid \ v \nsim \pivot\right\}.
\end{align*}
\end{comment}

\paragraph{Generalized poset sorting with comparable edges (GPSC).}
In this model, all edges in $G$ are comparable edges.
Specifically, when an edge in $E$ is queried, the answer will only be $v\prec u$ or $u \prec v$, corresponding to $\mathcal{P}$. We remark that when $\mathcal{P}$ is a total order set, then the model draws back to the generalized sorting model, so GPSC is already a generalization of generalized sorting.

\paragraph{Weighted generalized (total order) sorting.}
In this model, the poset is total order, and the query graph is weighted by a weight function $w : E \to \mathbb{R}_{\geq 0}$.
We aim to minimize the sum of costs we pay to solve the GPS under this setting.
We evaluate our algorithms by the competitive ratio, which is the maximum ratio taken over all possible inputs, measured by the cost of the algorithm, denoted as $\alg$, divided by $\opt$ which is the sum of costs $\sum_{i}w(v_i, v_{i+1})$ where $(v_1, \ldots, v_n)$ is the total order defined by $\mathcal{P}$.

     \section{A General Framework for Generalized Poset Sorting}
\label{sec:framework}

In this section, we present our framework for GPS.
As mentioned, this framework consists of two steps: it first reduces GPS to finding a linear extension, and eventually reducing the task of finding linear extension to constructing a partition oracle.
We start with formally define the mentioned linear extension problem and the partition problem.
We establish two sets of technical lemmas that relate GPS with
linear extension \Cref{thm:linearexttoposet} and the partition problem \Cref{lem:part2LE-main,lem:part2LE-main2}, respectively.

\paragraph{Linear extension problem.} In linear extension problem, 
there is an underlying poset $\mathcal{P} = (V,\prec)$ and query graph $G=(V,E)$.
The algorithm receives $G$ as input and has access to an oracle, which accepts queries $u,v$ such that $(u,v)\in E$ and answers the relation of $u,v$ in $\mathcal{P}$. The algorithm needs to compute a linear extension of $\mathcal{P}$ by making queries as few as possible. We call $p_1,p_2,\ldots ,p_n \in \perm(V)$ a linear extension of $\mathcal{P}$ if and only if for every $1\leq i<j\leq n$, $p_i\nsucc p_j$ (i.e. either $p_i \prec p_j$ or $p_i\incom p_j$). 

\begin{lemma}[Linear Extension to Poset]
\label{thm:linearexttoposet}
There exists an algorithm that given a linear extension of the underlying width-$k$ poset $\mathcal{P}$ solves GPS in $\tilde{O}(nk)$ queries.
\end{lemma}

\Cref{thm:linearexttoposet} shows that given any linear extension of $\mathcal{P}$, we can solve $\mathcal{P}$ using $\tilde{O}(nk)$ queries. Specifically, our algorithm maintains a vertex set $X$ such that all directions of edges in $G[X]$ is determined. Initially, $X$ is empty. The vertices are added to $X$ by their order in the linear extension. When vertex $v$ is added to $X$, our algorithm needs to determine all directions of edges between $X$ and $v$. In the proof of \Cref{thm:linearexttoposet}, we show that $X$ can always be decomposed into at most $k$ paths. For every vertex $v$ and every path $P=p_1\prec p_2\prec \ldots \prec p_s$, there exists $a,b$ such that
\begin{itemize}
    \item For every $1\leq i\leq a$, $p_i\prec v$. 
    \item For every $a<i<b$, $p_i\incom v$. 
    \item For every $b\leq i\leq s$, $p_i\succ v$. 
\end{itemize}

$a,b$ can be solved by binary search. Hence the directions of edges between $X$ and $v$ can be determined in $O(k\log n)$ probes. Our algorithm is shown as \Cref{alg:LEtoPoset}. 

\begin{algorithm}
\caption{Constructing GPS using Linear Extension}
\label{alg:LEtoPoset}
\begin{algorithmic}[1]
\Procedure{GPS}{$p$}
\State input: $p$, a linear extension of $\mathcal{P}$
\For{$i\gets 1,2,\ldots ,n-1$}
    \State let $X_i$ be $\{p_1,p_2,\ldots ,p_i\}$
    \State compute a path cover of induced subgraph $\vec{G}[X_i]$, denoted by $P_1,P_2,\ldots ,P_{w_i}$ ($w_i$ denotes the width of poset $\mathcal{P}(\vec{G}[X_i])$)
    \State determine the directions of edges between $X_i$ and $p_{i+1}$ by applying binary search between $P_j$ and $p_{i+1}$ for every $1\leq j\leq w_i$
\EndFor
\EndProcedure
\end{algorithmic}
\end{algorithm}

\begin{lemma}\label{lem:LE-prefix}
    For every linear extension $p_1,p_2,\ldots ,p_n$ of $\mathcal{P}$ and every $1\leq i\leq n$, let $X_i=\{p_1,p_2,\ldots ,p_i\}$,
    $\mathcal{P}(\vec{G}[X_i])=(X_i,\prec)$. 
\end{lemma}

\begin{proof}
    To prove $\mathcal{P}(\vec{G}[X_i])=(X_i,\prec)$, we show that for every $u,v\in X_i$, $u\prec v$ if and only if $u$ can reach $v$ in $\vec{G}[X_i]$. 
    \begin{itemize}
        \item If $u$ can reach $v$ in $\vec{G}[X_i]$, then $u$ can also reach $v$ in $\vec{G}$, which implies $u\prec v$. 
        \item If $u\prec v$, then there exists a path $u\to q_1 \to q_2 \to \ldots \to q_\ell \to v$ in $\vec G$. For every $1\leq i\leq \ell$, suppose $q_i \notin X_i$, let $q_i = x_a, v=x_b$, we have $x_a\prec x_b$ and $b \leq i<a$, which contradicts the definition of linear extension. Hence we have $q_1,q_2,\ldots ,q_\ell \in X_i$, which implies $u$ can reach $v$ in $\vec{G}[X_i]$. 
    \end{itemize}
\end{proof}

\begin{proof}[Proof of \Cref{thm:linearexttoposet}]
By \Cref{lem:LE-prefix}, for every $1\leq i<n$, $\mathcal{P}(\vec{G}[X_i])=(X_i,\prec)$. The width of poset $(X_i,\prec)$ is no more than $k$ since $X_i\subseteq V$. 
 Hence, for every $i$, all directions of edges between $p_1,p_2,\ldots ,p_i$ and $p_{i+1}$ can be determined by applying at most $k$ binary searches on the path cover of $\vec{G}[X_i]$.
Clearly, this entire process takes $O(nk\log n)$ queries in total, which finishes the proof. 
\end{proof}

\paragraph{Partition problem.} The partition algorithm is defined on an underlying DAG $\vec{G} = (V,\vec{E})$ and a query graph $G=(V,E)$. Notice that partition algorithm is a pure graph problem (there is no poset in the problem definition). The algorithm receives $G$ and vertex $p\in V$ as input and has access to an oracle, which accepts queries $u,v$ such that $(u,v)\in E$ and answers the relation of $u,v$ in $\vec{G}$. There are three possible relations, $u$ can reach $v$, $v$ can reach $u$, neither $u$ nor $v$ can reach each other. 
The algorithm needs to compute $\Lset{V}{p}, \Mset{V}{p}, \Rset{V}{p}$ by making queries as few as possible, where
\begin{itemize}
    \item $\Lset{V}{p}=\{u\in V\mid u\neq p,u\text{ can reach }p\}$
    \item $\Rset{V}{p}=\{u\in V\mid u\neq p,p\text{ can reach }u\}$
    \item $\Mset{V}{p}=\{u\in V\mid u\neq p,\text{neither $u$ nor $p$ can reach each other}\}$
\end{itemize} 

The following two lemmas reduce the problem of finding linear extensions to finding a partition of the elements with respect to a given pivot.
These two versions of lemmas are essentially the same, except that one allows the partition oracle to make mistakes and with a randomized query complexity (but needs to succeed with high probability),
and the other requires the correctness (with probability $1$) and a good query complexity in expectation.
We need these two since we find it is not trivial to convert one to the other,
and our downstream algorithms may need both of them.

\begin{lemma}\label{lem:part2LE-main}
    If for every $X\subseteq V$ and $p\in X$, 
    $\algPart(G[X], p)$ correctly outputs $\Lset{X}{p},\Mset{X}{p},\Rset{X}{p}$ within $O(|X|(k_{X_{\prec v}}+k_{X_{\succ v}})f(n,k))$ queries with probability $1-\varepsilon$ ($f$ is some function of $n,k$),
    then $\algLE(V)$ outputs a linear extension of $\mathcal{P}$ in $O(nkf(n,k)\log^2 n)$ queries with probability of at least $1-n\varepsilon - 2n^{-6}$. 
\end{lemma}

\begin{lemma}\label{lem:part2LE-main2}
    If for every $X\subseteq V$ and $p\in X$, 
    $\algPart(G[X], p)$ always correctly outputs $\Lset{X}{p},\Mset{X}{p},\Rset{X}{p}$ and uses $O(|X|(k_{X_{\prec v}}+k_{X_{\succ v}})f(n,k))$ queries in expectation,
    then $\algLE(V)$ outputs a linear extension of $\mathcal{P}$ in $O(nkf(n,k)\log^2 n)$ queries in expectation. 
\end{lemma}

The proof of \Cref{lem:part2LE-main} is left to \Cref{sec:proof_lem:part2LE-main}. \Cref{lem:part2LE-main2} can be proved by slightly modifying the proof of \Cref{lem:part2LE-main}. 
	
\subsection{Proof of \Cref{lem:part2LE-main}}
\label{sec:proof_lem:part2LE-main}

To prove \Cref{lem:part2LE-main}, we give a quicksort-like algorithm for computing linear extension.
	Roughly speaking, we start with picking a random element as a pivot, and then try to \emph{partition} the input into three parts according to how the element compare with the pivot.
	These three parts can be solved recursively and combining them yields a linear extension.
The detailed algorithm is shown in \Cref{alg:part-to-LE}.
Analysis of this quicksort-like algorithm for the special case of complete graph $G$ was given by~\cite{DBLP:journals/siamcomp/DaskalakisKMRV11}.
While it is possible to adapt their proof to our case in a straightforward way,
unfortunately, it introduces an additional factor $k$ on top of the complexity of the partition algorithm.
We employ a more careful analysis in the generalized setting, and eventually we can save this factor $k$, achieving an improved complexity.

\begin{comment}
    \begin{lemma}\label{lem:part2LE-main}
    If for every $X\subseteq V$ and $p\in X$, 
    $\algPart(G[X], p)$ outputs $L_{X,p},M_{X,p},R_{X,p}$ in $O(|X|(k_{X_{\prec v}}+k_{X_{\succ v}})f(n,k))$ queries with probability $1-\delta$ ($f$ is some function of $n,k$),
    where
   \begin{itemize}
    \item $L_{X,p}=\{u\in X\mid u\neq p,u\text{ can reach }p\}$
    \item $R_{X,p}=\{u\in X\mid u\neq p,p\text{ can reach }u\}$
    \item $M_{X,p}=X\backslash (L_{X,p}\cup R_{X,p}\cup \{p\})$
\end{itemize} 
    then $\algLE(V)$ outputs a linear extension of $\mathcal{P}$ in $O(nkf(n,k)\log^3 n)$ queries with probability of at least XXX. 
\end{lemma}
\end{comment}

Notice that the partition algorithm only receives the induced graph $G[X]$ as input, instead of having access to the entire graph. This is due to the recursive nature of the algorithm, where we wish to solve the subproblem entirely inside an induced subgraph.
Moreover, it only finds sets $\Lset{X}{p}, \Mset{X}{p}, \Rset{X}{p}$ (instead of $X_{\prec p},X_{\incom p},X_{\succ p}$). 
This could cause issues,
since ideally, if $\mathcal{P}(\vec{G}[X])=(X,\prec)$, then we have $\Lset{X}{p}=X_{\prec p},\Mset{X}{p}=X_{\incom p},\Rset{X}{p}=X_{\succ p}$ as we expected, but this does not hold for every induced subgraph $G[X]$ ($u\prec v$ does not imply $u$ can reach $v$ in $\vec{G}[X]$).
Luckily, we can show in \Cref{lem:induced-subgraph-issue} that $\mathcal{P}(\vec{G}[X])=(X,\prec)$ always holds in every recursive call,
and this makes sure even assuming such weaker input and output of the partition algorithm still works.

    \begin{algorithm}
\caption{Partition to Linear Extension}
\label{alg:part-to-LE}
\begin{algorithmic}[1]
\Function{Part-to-LE}{$X$}
\If{$X=\varnothing$}
	\State \textbf{return} an empty sequence
\EndIf
\State randomly select a pivot vertex $p$ from $X$
\State $L,M,R\gets \textsc{Partition}(G[X],p)$\Comment{the return value are expected to be $X_{\prec p},X_{\incom p},X_{\succ p}$ respectively}
\State $\mathrm{LE}\gets \algLE(L)||p||\algLE(M)||\algLE(R)$ \Comment{$A||B$ denotes the concatenation of $A,B$}
\State \textbf{return} $\mathrm{LE}$
\EndFunction
\end{algorithmic}
\end{algorithm}

To prove \Cref{lem:part2LE-main}, we need a tree structure called ternary search tree, which is used for query complexity analysis in~\cite{DBLP:journals/siamcomp/DaskalakisKMRV11}. 

\begin{definition}[Ternary Search Tree]
	A ternary search tree $T$ of poset $\mathcal{P}$ is defined as below. 
	
	\begin{itemize}
		\item $T$ consists of $n$ nodes. Each node corresponds to a vertex in $V$. We name each node by its corresponding vertex for convenience. 
        \item Let $X(v)$ denotes vertices in the subtree of $v$. For convenience, we omit the $v$ in subscript when we write $X(v)_{\prec v}$, i.e. we write $X(v)_\prec$ instead, the same for $X(v)_\incom, X(v)_\succ$. 
        \item The tree structure is defined in a recursive way, similar to $\algLE$. 
        \begin{itemize}
            \item[$\circ$] The root node is picked arbitrarily from $V$. We have $X(\Troot)=V$. 
            \item[$\circ$] Each node $v$ has at most three children, denoted by $\chL(v), \chM(v), \chR(v)$, which are picked arbitrarily from $X(v)_\prec,X(v)_\incom,X(v)_\succ$ respectively. We have $X(\chL(v))=X(v)_\prec,X(\chM(v))=X(v)_\incom,X(\chR(v))=X(v)_\succ$. If $X(v)_\prec=\varnothing$ then $v$ does not have a $\prec$-child, and the same for $X(v)_\incom, X(v)_\succ$. 
        \end{itemize}
        \item Denote the set of all possible ternary search trees of poset $\mathcal{P}$ by $\mathcal{T}_\mathcal{P}$. 
	\end{itemize}
\end{definition}

After defining the ternary search tree, we are ready to settle the issue mentioned before. 
\begin{lemma}\label{lem:induced-subgraph-issue}
    For every ternary search tree $T$ and $p\in V$, we have $\mathcal{P}(\vec{G}[X(p)])=(X(p),\prec)$. 
\end{lemma}

\begin{proof}
Fix some $p\in V$. 
We prove $\mathcal{P}(\vec{G}[X(p)])=(X(p),\prec)$ by showing that for every $u,v\in X(p)$, $u\prec v$ if and only if $u$ can reach $v$ in $\vec G[X(p)]$. 

For every $u,v\in X(p),u\neq v$ such that $u$ can reach $v$ in $\vec{G}[X(p)]$, $u$ can also reach $v$ in $\vec{G}$, which implies $u\prec v$. 

    For every $u,v\in X(p),u\prec v$, there exists a path $u \to x_1\to x_2\to \ldots \to x_s \to v$ in $\Gbase$. We show that $X$ contains all vertices in this path, which directly implies $u$ can reach $v$ in $\vec{G}[X(p)]$. 
    For every ancestor $y$ of $p$ in $T$, $u,v$ are in the same subtree of $y$ as $p$, hence there are three possible cases. 
    \begin{itemize}
        \item $u\prec y,v\prec y$. By $v\prec y$ we have $\forall 1\leq i\leq s,x_i\prec y$. 
        \item $u\succ y,v\succ y$. By $v\succ y$ we have $\forall 1\leq i\leq s,x_i\succ y$. 
        \item $u\incom y,v\incom y$. For every $x_i$, $x_i\prec y$ implies $u\prec y$, $x_i\succ y$ implies $v\succ y$, hence $x_i$ must be incomparable with $y$. 
    \end{itemize}

    For every $1\leq i\leq s$ and every ancestor $y$ of $p$, $x_i$ is in the same subtree of $y$ as $p$. This implies $X(p)$ contains $x_1,x_2,\ldots ,x_s$, which concludes our proof. 
\end{proof}

\Cref{lem:induced-subgraph-issue} implies that for every $T$ and $p\in V$, $X(p)_{\prec p}=\Lset{X(p)}{p}$, $X(p)_{\incom p}=\Mset{X(p)}{p}$, $X(p)_{\succ p}=\Rset{X(p)}{p}$. Now we can say that the ternary search tree is defined in exactly the same recursive way as $\algPart$. For each tree node, there is a corresponding recursive call $\algLE(X(p))$ in which $p$ is selected as pivot vertex, vice versa. 
\begin{comment}
\begin{corollary}\label{clry:bijection}
    There exists a bijection between the divide and conquer process of $\algLE$ and ternary search trees of $\mathcal{P}$. To be specific, we claim that
    \begin{itemize}
        \item For each divide and conquer process of $\algLE$, we construct a tree as follows. For every recursion call ``$\algLE(X)$ calls $\algLE(Y)$'', let $p_X,p_Y$ be the pivot vertices of $X,Y$, then we let $p_X$ be the parent node of $p_Y$. This tree is a ternary search tree. 
        \item For each ternary search tree $T$, there exists a divide and conquer process such that, every time $\algLE(X')$ is called for some $X'$, there exists $p\in V$ such that $X'=X(p)$ and $p$ is selected as the pivot vertex in $\algLE(X')$. 
    \end{itemize}
\end{corollary}

\Cref{clry:bijection} is actually implied by \Cref{lem:induced-subgraph-issue}, 
because $\algLE(X)$ recursively calls $\algLE(L_{X,p})$, $\algLE(M_{X,p})$, $\algLE(R_{X,p})$, but in ternary search tree the three subtrees of children of $p$ is $X(p)_{\prec p}$, $X(p)_{\incom p}$, $X(p)_{\succ p}$. The statement of \Cref{clry:bijection} becomes an assumption in the following discussion and will not be mentioned explicitly. 
\end{comment}

Before proving \Cref{lem:part2LE-main}, there is one more issue for random graphs. In \Cref{sec:random_graph}, we propose a partition algorithm for random graphs. However, this is based on the assumption that the input graph, i.e. the induced subgraph $G[X(p)]$, is a random graph under our model. We show that this assumption always holds in \Cref{lem:random-subgraph}.

\begin{lemma}\label{lem:random-subgraph}
    For every ternary search tree $T\in \mathcal{T}_\mathcal{P}$, every $v\in V$ and every minimal DAG $\Gbase$ such that $\mathcal{P}(\Gbase) = \mathcal{P}$, 
    let $G$ be a random query graph under \Cref{def:er}, then
    $G[X(v)]$ is also a random query graph. 
\end{lemma}

\begin{proof}
    To show that $G[X(v)]$ is a random query graph, we verify three requirements of random query graph in \Cref{def:er}.
    \begin{itemize}
        \item Notice that $\uGbase$ also satisfies the requirement of input query graph $G$, i.e. $\mathcal{P}(\Gbase) = \mathcal{P}$, which means we can also apply \Cref{lem:induced-subgraph-issue} on $\Gbase$ to imply $\mathcal{P}(\Gbase[X(v)]) = (X(v), \prec)$. 
        \item For every edge in $\vec{G}[X(v)]$, since it is not a redundant edge in $\vec{G}$, it must not be a redundant edge in $\vec{G}[X(v)]$. Hence $\vec{G}[X(v)]$ is a minimal DAG. 
        \item By \Cref{def:er}, $G$ is a union of $\uGbase$ and $G(n,p)$ (here $p$ is a parameter of random graph), then $G[X(v)]$ is also a union of $\uGbase[X(v)]$ and $G(|X(v)|, p)$. 
    \end{itemize}
    This concludes our proof. 
\end{proof}

The statement of \Cref{lem:part2LE-main} can be separated into two parts, correctness part (i.e. $\algLE(V)$ outputs a linear extension of $\mathcal{P}$) and efficiency part (i.e. $\algLE(V)$ ends within $\tilde{O}(nkf(n,k))$ queries). We prove the correctness part first. 

In the correctness proof, we expect $\algLE(V)$ to output a linear extension of $\mathcal{P}$. For other vertex sets $X\subset V$, we expect $\algLE(X)$ to output a vertex sequence $(\ell_1,\ell_2,\ldots ,\ell_{|X|}) \in \perm(X)$ such that for every $1\leq i<j\leq |X|$, $\ell_i \nsucc \ell_j$. These sequences are called linear extensions of $X$. 

\begin{lemma}\label{lem:LE-recur}
For every $X\subseteq V$ and $p\in X$, let $L_\prec,L_\incom,L_\succ$ be some linear extension of $X_{\prec p},X_{\incom p},X_{\succ p}$ respectively, then $L = L_\prec || p || L_\incom || L_\succ$ is a linear extension of $X$ (here $A||B$ denotes the concatenation of sequence $A,B$, $p$ represents a sequence that only contains vertex $p$). 
\end{lemma}

\begin{proof}
    Let $L = (\ell_1,\ell_2,\ldots ,\ell_{|X|})$. For every $1\leq i<j\leq |X|$, there are four cases. 
    \begin{enumerate}
        \item $\ell_i = p$ or $\ell_j = p$: $\ell_i \nsucc \ell_j$ can be derived directly from the definition of $L$. 
        \item $\ell_i\in L_\prec, \ell_j \in L_\succ$: We have $\ell_i \prec p \prec \ell_j$. 
        \item $\ell_i \in L_\prec,\ell_j\in L_\incom$: $\ell_i \succ \ell_j$ implies $\ell_j \prec \ell_i \prec p$, which contradicts $\ell_j\in L_\incom$. Hence we have $\ell_i\nsucc \ell_j$. 
        \item $\ell_i\in L_\incom,\ell_j\in L_\succ$: Similar to Case 3. 
    \end{enumerate}

    In summary, for every $1\leq i<j\leq |X|$, $\ell_i\succ \ell_j$ holds. This implies $L$ is a linear extension of $X$. 
\end{proof}

\begin{lemma}[Correctness]\label{lem:LE-correctness}
    If for every $X\subseteq V$ and $p\in X$, 
    $\algPart(G[X], p)$ outputs $\Lset{X}{p},\Mset{X}{p},\Rset{X}{p}$ with probability $1-\varepsilon$, then $\algLE(V)$ outputs a linear extension of $\mathcal{P}$ with probability of at least $1-n\varepsilon$. 
\end{lemma}

\begin{proof}
    Let $\mathcal{E}_{X,p}$ be the event that $\algPart(G[X],p)$ outputs $\Lset{X}{p},\Mset{X}{p},\Rset{X}{p}$. 
    $$\Pr_{T\sim \mathcal{T}_{\mathcal{P}}}[\forall p\in V,\mathcal{E}_{X(p),p}\text{ happens}] \geq \sum_{T\in \mathcal{T}_\mathcal{P} }\Pr[T]\left(1-\sum_{p\in V}\Pr[\mathcal{E}_{X(p),p}\text{ does not happen}] \mid T]\right) \geq 1-n\varepsilon$$

    Next we prove for every $T$, conditioning on $\forall p\in V,\mathcal{E}_{X(p),p}$ happens, for every $p\in V$, $\algLE(X(p))$ outputs a linear extension with probability 1. 

    Fix some $T$. We prove by induction on the tree structure $T$ that $\algLE(X(p))$ outputs a linear extension of $X(p)$. We denote $X(p)$ by $X$ for convenience. 
    
    \paragraph{Base case.}
    If $p$ is a leaf node of $T$, then $X_{\prec p},X_{\incom p},X_{\succ p}$ are all empty. Hence $\algPart(X,p)=p$, which is a linear extension of $X=\{p\}$. 
    
    \paragraph{Inductive case.}
    Otherwise, let $L_\prec =\algLE(X_{\prec p})$, $L_\incom=\algLE(X_{\incom p})$, $L_\succ =\algLE(X_{\succ p})$.
    By induction hypothesis, $L_\prec ,L_\incom,L_\succ$ are linear extensions of $X_{\prec p},X_{\incom p},X_{\succ p}$, respectively. By Lemma \ref{lem:LE-recur}, $\algLE(X) = L_\prec || p || L_\incom || L_\succ$ is a linear extension of $X$. 
\end{proof}

We start our query complexity analysis with a weaker bound. Let $\#\text{queries}$ denotes the total number of queries used by $\algLE(V)$. 
Conditioning on $\algPart$ always produces the correct output, we have
\begin{align*}
    \#\text{queries} &\leq f(n,k)\sum_{p \in V}(k_{X(p)_{\prec}}+k_{X(p)_{\succ}})|X(p)| \\
    &\leq f(n,k)\cdot k\sum_{p\in V}|X(p)| \\
    &=f(n,k)\cdot k\sum_{v\in V}(|\Anc(v)|+1)
\end{align*}
where $\Anc(v)$ denotes the set of ancestors of $v$ in $T$ ($v$ itself is not included).

We divide $\Anc(v)$ into the following two type of ancestors and give upper bounds for them respectively. 
\begin{itemize}
	\item Ancestors comparable with $v$ are called comparable ancestors, denoted by $\CAnc(v)$. 
    \item Ancestors incomparable with $v$ are called incomparable ancestors, denoted by $\IAnc(v)$. 
\end{itemize}

\begin{lemma}\label{lem:AAnc}
	For every ternary search tree $T$ and $v\in V$, $|\IAnc(v)|<k$. 
\end{lemma}

\begin{proof}
	For any vertex pair $x,y$ in $\IAnc(v)\cup \{v\}$, suppose $x$ is the ancestor of $y$, then $y$ is in the $\incom$-subtree of $x$, i.e. $x\incom y$. Hence $\IAnc(v)\cup \{v\}$ contains at most $k$ vertices since it is an antichain of poset $\mathcal{P}$. 
\end{proof}

\begin{restatable}{lemma}{CAncUB}\label{lem:CAnc}
	For every $v\in V$, $|\CAnc(v)|\leq 60\log n$ holds with probability at least $1-n^{-10}$. 
\end{restatable}
\begin{proof}
    The detailed proof can be found in \Cref{sec:proof_lem:CAnc}.
\end{proof}

The proof of \Cref{lem:CAnc} is based on the following observation. For every $X\subseteq V$, let $x_1,x_2,\ldots ,x_\ell$ be a linear extension of $X$. We have $X_{\prec x_i}\subseteq \{x_1,x_2,\ldots ,x_{i-1}\}$, $X_{\succ x_i}\subseteq \{x_{i+1},x_{i+2},\ldots ,x_\ell\}$. Roughly speaking, the size of $X_{\prec p}$ and $X_{\succ p}$ is about half of $|X|$. The remaining deduction is similar to the analysis of quicksort. We leave the full proof in appendix for novelty. 

Combining \Cref{lem:LE-correctness}, \Cref{lem:AAnc} and \Cref{lem:CAnc}, we have $\nQuery\leq f(n,k)\cdot k\sum_{v\in V}(|\Anc(v)|+1)=O(nk^2 f(n,k))$ with high probability. We want to improve this bound to $\tilde{O}(nk f(n,k))$, so that our algorithm is nearly optimal when $f(n,k)=\tilde{O}(1)$. 

Recall we assumed that the partition algorithm uses $O((k_{X(p)_\prec}+k_{X(p)_\succ})|X(p)|f(n,k))$ queries. For some $p\in V$, $k_{X(p)_\prec}+k_{X(p)_\succ}$ can be much smaller than $k_\mathcal{P}$. We capture this key property by the following lemma. The proof is left to next section. 

\begin{lemma}\label{lem:AAnc-k}	
 For every $v\in V$, $\sum_{p\in \IAnc(v)}(k_{X(p)_\prec}+k_{X(p)_\succ })=O(k\log^2 n)$ with probability at least $1-n^{-7}$. 
\end{lemma}
Now we are ready to prove our main lemma. 

\begin{proof}[Proof of \Cref{lem:part2LE-main}]
Conditioning on $\algPart$ always produces the correct output, we have
\begin{align*}
	\#\mathrm{queries} &=
	\sum_{p\in V}\left(k_{X(p)_\prec}+k_{X(p)_\succ}\right)|X(p)| f(n,k)\\
	&\leq f(n,k)\sum_{v\in V}\left(k+\sum_{p\in \Anc(v)}\left(k_{X(p)_\prec}+k_{X(p)_\succ}\right)\right) \\
	&\leq f(n,k)\sum_{v\in V}\left(k\cdot (|\CAnc(v)|+1)+\sum_{p\in \IAnc(v)}\left(k_{X(p)_\prec}+k_{X(p)_\succ}\right)\right)
\end{align*}

Combining \Cref{lem:LE-correctness}, \Cref{lem:CAnc} and \Cref{lem:AAnc-k}, $\algLE(V)$ outputs a linear extension of $\mathcal{P}$ in $O(nkf(n,k)\log^2 n)$ queries with probability of at least $1-n\varepsilon - 2n^{-7}$.

\end{proof}

\begin{comment}
In , we will introduce a partition algorithm for random graph. Here we want to highlight that both Algorithm \ref{alg:part-to-LE} and the ternary search tree are defined on the poset $\mathcal{P}$ and are not affected by the comparison graph $G$. 
\end{comment}

\subsection{Proof of Lemma \ref{lem:AAnc-k}}

Fix $v$ and the ancestors of $v$, denoted by set $P$, i.e. $P=\Anc(v)$. 
Let the path from $\Troot$ to $v$ (in tree $T$) be $p_1\to p_2\to \ldots \to p_d \to v$ ($p_1=\Troot$). $\mathbf{p} = (p_1,p_2,\ldots ,p_d)$ is a random permutation of $P$ (not uniformly at random). It is natural to consider which permutations of $P$ are possible values of $\mathbf{p}$, denoted by set $\mathrm{PathPerm}_v(P)$. Define 
$$\sgn(x,y)=\begin{cases}
	-1 & x\prec y \\
	+1 & x\succ y \\
	0 & x\incom y
\end{cases}$$
If $x$ is ancestor of $y$ in the ternary search tree, then $\sgn(x,y)$ indicates which subtree $y$ is in ($\prec,\incom$ or $\succ$). Then for every $1\leq i<d$, we can define $\tilde{X}(q_1,q_2,\ldots ,q_i)$ as the set $X(p_{i+1})$, which is uniquely determined by $p_1=q_1,p_2=q_2,\ldots ,p_i=q_i$. The formal definition is
$$\tilde{X}(q_1,q_2,\ldots ,q_i)=\{u\in V\mid \forall 1\leq j\leq i,\sgn(q_j,u)=\sgn(q_j,v)\}.$$
In particular, let $\tilde{X}(\text{empty list})=V$, which corresponds to the special case of $i=0$. 

Now we can give a formal definition of $\mathrm{PathPerm}_v(P)$ by the following lemma, which is equivalent to the previous definition. 
\begin{lemma}\label{lem:pathperm-form}
    For every $\mathbf{q}\in \perm(P)$, $\mathbf{q}\in \pperm_v(P)\iff \forall 1\leq i<j\leq d,\sgn(q_i,q_j)=\sgn(q_i,v)$. 
\end{lemma}
\begin{proof}
    For every $\mathbf{q}\in \pperm_v(P)$ and $1\leq i<j\leq d$, $v$ and $q_j$ are in the same subtree of $q_i$, hence we have $\sgn(q_i,q_j)=\sgn(q_i,v)$. 

    For every $\mathbf{q}\in \perm(P)$ such that $\forall 1\leq i<j\leq d,\sgn(q_i,q_j)=\sgn(q_i,v)$, we can verify that for every $1\leq i\leq d,q_i\in \tilde{X}(q_1,q_2,\ldots ,q_{i-1})$. 
    By selecting $q_i$ as the pivot vertex of $\tilde{X}(q_1,q_2,\ldots ,q_{i-1})$ for every $1\leq i\leq d$, we can construct a ternary search tree such that $p_1=q_1,p_2=q_2,\ldots ,p_d = q_d$, which implies $\mathbf{q} \in \pperm_v(P)$. 
\end{proof}

For every $x,y\in P, x\neq y$, $\sgn(x,y),\sgn(x,v)$ and $\sgn(y,v)$ are all determined, only the relative order of $x,y$ in $p$ is unknown. By \Cref{lem:pathperm-form}, for some of the vertex pairs $x,y\in P$, it is possible to determine the relative order of $x,y$ in $\mathbf{p}$ according to $\sgn(x,y),\sgn(x,v)$ and $\sgn(y,v)$. 
In specific, there are three possible cases. 

\paragraph{Both $x,y$ are comparable ancestors.} Let set $C$ denotes the comparable ancestors of $v$, i.e. $\CAnc(v) = C$. We sort the comparable ancestors by the order they occur in $\mathbf{p}$, denoted by sequence $c_1,c_2,\ldots ,c_\ell$ ($\ell = |C|$). We first notice that $(C,\prec)$ forms a total order. This is because $\forall 1\leq i<j\leq \ell$, we have $\sgn(c_i,c_j)=\sgn(c_i,v)\neq 0$, i.e. $c_i,c_j$ is comparable. Moreover, we observe that for every $1\leq i\leq \ell$, $c_i$ is either the minimum element or the maximum element of $\{c_i,c_{i+1},\ldots ,c_\ell\}$. 

\paragraph{$x$ is comparable ancestor while $y$ is not.} By \Cref{lem:pathperm-form}, we can verify that $\sgn(x,y)\neq 0 \iff \mathrm{rank}_\mathbf{p}(x)<\mathrm{rank}_\mathbf{p}(y)$. 
If we fix some $z\in \IAnc(v)$, then all comparable ancestors that are comparable with $z$ should be placed before $z$ in $\mathbf{p}$ while all comparable ancestors that are incomparable with $z$ should be placed after $z$ in $\mathbf{p}$.
Combining this result with the observation in previous case, we have the following corollary. For every incomparable ancestor $z\in \IAnc(v)$, let $i=|C_{\prec z}|,j=|C_{\succ z}|$, then $z$ should be placed in some position between $c_{i+j}$ and $c_{i+j+1}$ in $\mathbf{p}$, i.e. $\rnk_\mathbf{p}(c_{i+j})< \rnk_\mathbf{p}(z)<\rnk_\mathbf{p}(c_{i+j+1})$. Moreover, we have $\{c_1,c_2,\ldots ,c_{i+j}\}=C_{\prec z}\cup C_{\succ z}$. 

\paragraph{Both $x,y$ are incomparable ancestors.} Just like we discussed before, $\IAnc(v)$ forms an antichain of $\mathcal{P}$, i.e. we always have $\sgn(x,y)=0$. In this case, both $\rnk_\mathbf{p}(x)<\rnk_\mathbf{p}(y)$ and  $\rnk_\mathbf{p}(x)>\rnk_\mathbf{p}(y)$ are possible. 

According to the discussion above, we fix $c_1,c_2,\ldots ,c_\ell$, i.e. the comparable ancestors of $v$ and the order they occur in $\mathbf{p}$, such that
\begin{itemize}
    \item For every $1\leq i\leq \ell$, $c_i$ is either the minimum element or the maximum element of $\{c_i,c_{i+1},\ldots ,c_\ell\}$. 
    \item For every $z\in Z$, $\{c_1,c_2,\ldots ,c_t\}=C_{\prec z}\cup C_{\succ z}$, where $t=|C_{\prec z}|+|C_{\succ z}|$. 
\end{itemize}
If no such order exists, then $\Anc(v)$ must not be $P$. For incomparable ancestors, we divide them into $\ell+1$ classes $Z_0,Z_1,\ldots ,Z_\ell$ according to $|C_{\prec z}|+|C_{\succ z}|$, i.e. $Z_i:=\{z\in Z \mid |C_{\prec z}|+|C_{\succ z}|=i\}$. For every $0\leq i\leq \ell$, vertices in $Z_i$ must be inserted between $c_i$ and $c_{i+1}$. In particular, vertices in $Z_0$ must be inserted before $c_1$, vertices in $Z_\ell$ must be inserted after $c_\ell$. 
For every $1\leq i\leq \ell$, let $r_i=i +\sum_{0\leq j<i} |Z_j|$, $c_i$ must be placed in the $r_i$-th position in $\mathbf{p}$, i.e. $p_{r_i}=c_i$. In particular, let $r_0=0,r_{\ell+1}=d+1$. 

\begin{lemma}\label{lem:perm-prob}
    For every $\hat{Z}_0\in \mathrm{perm}(Z_0),\ldots ,\hat{Z}_\ell\in \mathrm{perm}(Z_\ell)$, let $(q_1,q_2,\ldots ,q_d)=\hat{Z}_0 || c_1 || \hat{Z}_1 || \ldots || c_\ell || \hat{Z}_\ell$. We have
    $$\Pr[p_1=q_1,p_2=q_2,\ldots ,p_d=q_d \mid \Anc(v)=P,c_1,c_2,\ldots ,c_\ell] = \prod_{i=0}^\ell \frac{1}{|Z_i|!}.$$
\end{lemma}

\begin{proof}
    For every $1\leq i\leq d$, conditioning on $\Anc(v)=P$, $p_1=q_1,p_2=q_2,\ldots ,p_{i-1}=q_{i-1}$ and $c_1,c_2,\ldots ,c_\ell$, $p_i$ is uniformly distributed over some vertices in $\tilde{X}(q_1,q_2,\ldots ,q_{i-1})$, denoted by set $U_i$. For every $1\leq i\leq d$, let $j^*=\min \{ j \mid i<r_j \}-1$, we have
    $$U_i=\begin{cases}
        \{c_{j^*}\} & i = r_{j^*} \\
        Z_{j^*} \setminus \{ q_{i'} \mid r_{j^*-1} < i' < i\} & i \neq r_{j^*}
    \end{cases}$$
    For every $1\leq i\leq d$, let $\mathcal{E}_i$ be the event that $p_1=q_1,p_2=q_2,\ldots ,p_i=q_i$. 
    We can verify that 
    \begin{align*}
    \Pr[\mathcal{E}_d \mid \Anc(v)=P,c_1,c_2,\ldots ,c_\ell] &= \prod_{i=1}^d \Pr[\mathcal{E}_i \mid \Anc(v)=P,c_1,c_2,\ldots ,c_\ell,\mathcal{E}_{i-1}] \\
    &=E\left[\prod_{i=1}^d \frac{1}{|U_i|} \mid \Anc(v)=P,c_1,c_2,\ldots ,c_\ell \right] \\
    &=\prod_{i=0}^\ell \frac{1}{|Z_i|!}.
    \end{align*}
    
\end{proof}

\begin{proof}[Proof of Lemma \ref{lem:AAnc-k}]
    Fix some $0\leq i\leq \ell$. We show that $\sum_{z\in Z_i}(k_{X(z)_\prec} + k_{X(z)_\succ}) = O(k\log k)$ with high probability. 

    We sort the vertices in $Z_i$ by the order they occur in $\mathbf{p}$, denoted by $z_1,z_2,\ldots ,z_{t_i}$, where $t_i = |Z_i|$. By \Cref{lem:perm-prob}, $z_1,z_2,\ldots ,z_{t_i}$ is uniformly distributed over $\perm(Z_i)$. 

    Pick an arbitrary chain decomposition of $\mathcal{P}$, let it be $\mathcal{S}=\{S_1,S_2,\ldots ,S_k\}$. For each chain $S_i$, we write its vertices by ascending order $s_{i,1}\prec s_{i,2}\prec \ldots \prec s_{i,|S_i|}$. 
    Let $L_1(u),L_2(u),\ldots ,L_k(u)$ be integers such that $V_{\prec u}=\bigcup_{1\leq a\leq k}\{s_{a,b}\mid 1\leq b\leq L_i(u)\}$. 

    Recall that $z_1,z_2,\ldots ,z_{t_i}$ are all incomparable ancestors of $v$. Hence we have $X(z_j)_\prec \subseteq V\backslash \bigcup_{1\leq j'<j}X(z_{j'})_\prec$ for every $1\leq j\leq t_i$. 
In each chain $S_a$, we have 
$$X(z_j)_\prec\cap S_a \subseteq \left\{s_{a,b}\mid b \in \left(\max_{1\leq j'<j} L_a(z_{j'}) , L_a(z_j)\right] \right\}.$$ 
    That is to say, $X(z_j)_\prec\cup S_a \neq \varnothing$ only if $L_a(z_j) > \max_{1\leq j'< j} L_a(z_{j'})$. Then we have
$$\sum_{j=1}^{t_i}k_{X(z_j)_\prec} \leq \sum_{a=1}^k\sum_{j=1}^{t_i} \mathbb{I}(X(z_j)_\prec\cup S_a \neq \varnothing) \leq \sum_{a=1}^k\sum_{j=1}^{t_i} \mathbb{I}\left(L_a(z_j)>\max_{1\leq j'< j} L_a(z_{j'})\right).$$
    where $\mathbb{I}$ is the indicator function. 

    $L_a(z_1),L_a(z_2),\ldots ,L_a(z_{t_i})$ is a uniform permutation of $\{L_a(z_j) \mid 1\leq j\leq t_i\}$. Although $\{L_a(z_j) \mid 1\leq j\leq t_i\}$ may contain duplicative elements, we can break tie arbitrarily for the equal elements so that we can apply a well-known bound, stated as below. 
    
    \begin{restatable}{lemma}{SumMax}\label{lem:sum-max}
        Let $A$ be a set of $m$ real numbers. $p$ is a uniform permutation of $A$. 
    With probability at least $1-\varepsilon$ ($\varepsilon \leq 1/m$), we have
    $$\sum_{i=1}^m \mathbb{I}\left(p_i > \max_{1\leq i'<i} p_{i'}\right) \leq 3\ln(1/\varepsilon).$$
    \end{restatable}
    \begin{proof}
        The proof can be found in \Cref{sec:proof_lem:sum-max}.
    \end{proof}

    By applying union bound over $i=0,1,\ldots ,\ell, a=1,2,\ldots ,k$, together with \Cref{lem:sum-max} (take $\varepsilon = n^{-9}$), we prove that with probability of $1-n^{-7}$, 
$$\sum_{p\in \IAnc(v)}k_{X(p)_\prec} \leq \sum_{i=1}^\ell \sum_{j=1}^{t_i}\sum_{a=1}^k \mathbb{I}\left(L_a(z_j)>\max_{1\leq j'< j} L_a(z_{j'})\right) \leq 27k\ell \log n.$$

    We can prove that $\sum_{p\in \IAnc(v)} k_{X(p)_\succ} \leq 27k\ell \log n$ in the same way. 

    Combine this inequality with Lemma \ref{lem:CAnc} which says $\ell \leq 60\log n$ with probability at least $1-n^{-7}$, then apply union bound over all vertices $v$. 
    We conclude our proof. 
\end{proof}

\section{Partition Algorithms for \ER Query Graphs}
\label{sec:random_graph}

We prove \Cref{thm:er} in this section.
Due to the reductions introduced in \Cref{sec:framework} (\Cref{thm:linearexttoposet} and \Cref{lem:part2LE-main,lem:part2LE-main2}),
it suffices to design a partition algorithm for \ER query graphs.

\MainThmRG*

In our partition algorithm, given in \Cref{alg:partition-rand},
we do a graph traversal to identify the vertices smaller and larger to the pivot (and the remaining vertices are those not comparable to the pivot).
While this can be trivially done by a vanilla BFS, it suffers an efficiency issue as it needs to examine all edges in the graph (which is $O(n^2p)$ w.h.p.).
To resolve this issue, we propose a variant of BFS that can skip vertices in the queue.
Roughly speaking, we assign a health-point (HP) to every vertex with an initial value $R \geq 1$, and the HP of a vertex is decreased every time it is hit by a BFS exploration of other vertices. A dead vertex, i.e., whose HP reaches $0$, is skipped in the BFS queue and cannot explore its neighbors.
We show that for every parameter $R$, the complexity of this BFS is only $\tilde{O}(nRk)$.

In fact, the vanilla BFS may be viewed as $R = n$ case where no vertex is skipped.
On the other hand, we can achieve a better performance when $R$ is smaller, but we may not find the correct distances to the pivot.
Hence, it is crucial to find a suitable value of $R$.
To this end, we show that setting $R = \tilde{O}(k)$, which reduces the linear dependence in $n$ to $\poly\log n$ and the main structure parameter $k$ of the poset, actually yields the \emph{same} result (i.e., distances to pivot) as in $R = n$, with high probability.

Since our partition algorithm does not always produce the correct output, 
in order to satisfy the requirement of \Cref{lem:part2LE-main} (the lemma that translates partition oracle to linear extension),
the probability that our algorithm fails must be upper bounded by a parameter $\varepsilon$ that is independent of $n$.
To this end, we need to introduce another parameter $N$, which denotes the number of vertices of the original input graph in GPS problem. 
 
	\begin{algorithm}[h]
\caption{Partition Algorithm: Vertex-skipping BFS (\textsc{Skip-BFS})}\label{alg:partition-rand}
\begin{algorithmic}[1]
\Function{Skip-BFS}{$\pvt$}
\State $\hat{D}_0\gets \{\pvt\}$
\State $\hat{D}_1\gets \{\mathrm{in}(\mathrm{pivot})\}$
\State for every $\ell\geq 2$, initialize $\hat{D}_\ell$ as $\varnothing$

\State let parameter $R\gets k+18\log N$
\State for all $v\in V$, initialize $c[v]$ as $R$
\For{$\ell=1,2,\ldots ,n$}
    \State permute $\hat{D}_\ell$ in a random order
    \For{$v\in \hat{D}_\ell$ \label{line:org-for2}}
        \If{$c[v]>0$ \label{line:org-if}}
            \State for every vertex $u\in V\setminus \hat{D}_{\leq \ell-1}$, query edge $(u,v)$
            \State for every vertex $u\in \hat{D}_\ell,(u,v)\in \vec{E}$, decrease $c[u]$ by 1
            \State add all vertices $u\in V\setminus \hat{D}_{\leq \ell}$ such that $(u,v)\in \vec{E}$ to $\hat{D}_{\ell+1}$
\EndIf
    \EndFor
\EndFor
\State \textbf{return} $\bigcup_{\ell=1}^n \hat{D}_\ell$ \Comment{$\bigcup_{\ell=1}^n \hat{D}_\ell$ is expected to be $V_\prec$}
\EndFunction
\end{algorithmic}
\end{algorithm}

Here we define some notations for discussion. 
\begin{definition}
    For a DAG $\vec{G} = (V,\vec{E})$, for every vertex $v\in V$, define $\gIn_{\vec G}(v):=\{u\mid (u,v)\in \vec{E}\}$ as the predecessors of $v$, define $\gOut_{\vec G}(v):=\{u\mid (v,u)\in \vec{E}\}$ as the successors of $v$. 
    For every vertex $u,v\in V$, let $\vec{E}(u,v)$ denote the event that $(u,v)\in \vec{E}$. 
    For every vertex $u,v\in V$ such that $u\prec v$, define $d_{\vec G}(u,v)$ as the length of the shortest path from $u$ to $v$. 

    For an undirected graph $G=(V,E)$, for every vertex $v\in V$, define $\gAdj_G(v):=\{u \mid (u,v)\in V\}$. For every vertex $u,v\in V$, let $E(u,v)$ denote the event that $(u,v)\in E$. 

    For $0\leq \ell \leq n$, define $D_\ell:=\{v\mid d_{\vec G}(v,\pvt) = \ell\}$, where $\vec G$ is the underlying directed graph in GPS problem. Define $D_{\leq \ell} = \bigcup_{0\leq i\leq \ell} D_i$, $\hat{D}_{\leq \ell} = \bigcup_{0\leq i\leq \ell} \hat{D}_i$. 

    $V_\prec, V_\succ, k_\prec, k_\succ$ denotes $V_{\prec \pvt}, V_{\succ \pvt}, k_{V_{\prec \pvt}}, k_{V_{\succ \pvt}}$, respectively. 
\end{definition}
	
\subsection{Correctness}

In the correctness proof, we want to ensure that our algorithm does not miss any vertex of $V_\prec$ in the exploration. Ideally, for every vertex $v\in V_\prec$, $v$ should be explored when $\ell = d_{\vec G}(v,\pvt)$, i.e. $\forall v\in V_\prec, v\in \hat{D}_{d_{\vec G}(v,\pvt)}$. This is formally stated by \Cref{lem:rg-correctness}.

\begin{lemma}[Correctness] \label{lem:rg-correctness}
With probability $1-N^{-4}$, for every $1\leq \ell < n$, $\hat{D}_\ell = D_\ell$. 
\end{lemma}
	
\paragraph{Oracle model for random graph.} In our analysis, we need to prove the independence between certain random events and the randomness of edges $E(u,v)$. Although \Cref{alg:partition-rand} gives a clear description of Skip-BFS, it is hard to see whether a random event depends on a certain edge $E(u,v)$. 
To clearly demonstrate the independence between events and edges, we give equivalent descriptions of Algorithm \ref{alg:partition-rand} under an \emph{oracle model}. In this oracle model, the undirected graph $G=(V,E)$ is given by an oracle. Our algorithm may \textbf{reveal} the existence of edge $(u,v)$ by asking the oracle. 
	For every vertex pair, $E(u,v)$ is only allowed to be revealed once. If the value of some random variable $x$ is always determined before $E(u,v)$ is revealed, then  $x$ is independent of $E(u,v)$. 

To prove \Cref{lem:rg-correctness}, we give an equivalent description of \Cref{alg:partition-rand} under this oracle model in \Cref{alg:partition-rand-oracle}. The proof of \Cref{lem:rg-correctness} is based on \Cref{alg:partition-rand-oracle} and use the notations therein.

Below are some definitions and explanations about \Cref{alg:partition-rand-oracle}. 
\begin{itemize}
    \item There are three for-loops in \Cref{alg:partition-rand-oracle}. We use ``$\ell$-th iteration'' or iteration $\ell = x$ to refer to the for-loop in line \ref{line:oracle1-for1}, use ``$i$-th iteration'' or iteration $i = x$ to refer to the for-loop in line \ref{line:oracle1-for2}. We do not refer to a certain iteration of the for-loop in line \ref{line:oracle1-for3}. 
    \item In \Cref{alg:partition-rand}, we keep tracking on all HP counters throughout the $|\hat{D}_\ell|$ iterations (of the for-loop in line \ref{line:org-for2}). But these counters are only used in line \ref{line:org-if}. For each counter $c[v_i^{(\ell)}]$, only the value at the beginning of iteration $i$ is used to check whether $v_i^{(\ell)}$ should be skipped. We use $\hat{c}_i$ to record this value in 
 \Cref{alg:partition-rand-oracle}. 
    \item In the following discussion, the superscript $(\ell)$ is sometimes omitted when $\ell$ is clear. 

    \item The main difference between \Cref{alg:partition-rand} and \Cref{alg:partition-rand-oracle} is that the evaluation of $\hat{c}_i$ and $\hat{D}_{\ell+1}$ are delayed to the moment they are used. In oracle model, revealing the edges later helps our analysis. 
\end{itemize}

	\begin{algorithm}[h]
\caption{Skip-BFS under Oracle Model (used in \Cref{lem:rg-correctness})}\label{alg:partition-rand-oracle}
\begin{algorithmic}[1]
\Function{Oracle-Skip-BFS1}{$\pvt$}
\State $\hat{D}_0\gets \{\pvt\}$
\State $\hat{D}_1\gets \{\mathrm{in}(\mathrm{pivot})\}$
\State let parameter $R\gets k+18\log N$

\For{$\ell\gets 1,2,\ldots ,n$ \label{line:oracle1-for1}}
    \State let $v_1^{(\ell)},v_2^{(\ell)},\ldots ,v_{W_\ell}^{(\ell)}$ be a random permutation of $\hat{D}_\ell$
    \For{$i\gets 1,2,\ldots ,W_\ell$ \label{line:oracle1-for2}}
    	\State reveal and query all edges between $\left\{v_j^{(\ell)}\mid j<i,\hat{c}_j^{(\ell)}>0\right\}$ and $v_i^{(\ell)}$
    	\State compute $\hat{c}_i^{(\ell)} \gets R-\sum_{j\in [i-1],\hat{c}_j^{(\ell)}>0}\vec{E}(v_i^{(\ell)},v_j^{(\ell)})$\Comment{$\hat{c}_i^{(\ell)}$ does not change after it is computed}
    \EndFor
    \State define $I^{(\ell)}$ as $\{i:\hat{c}_i>0\}$
    \For{$i\in I^{(\ell)}$ \label{line:oracle1-for3}}
\State reveal and query all edges between $V\setminus \hat{D}_{\leq \ell}$ and $v_i^{(\ell)}$ \Comment{$\hat{D}_{\leq \ell}$ is defined as $\hat{D}_0 \cup \hat{D}_1 \cup \ldots \cup \hat{D}_\ell$}
\State add all vertices $u\in V\setminus \hat{D}_{\leq \ell}$ such that $(u,v_i^{(\ell)})\in \vec{E}$ to $\hat{D}_{\ell+1}$
	\EndFor
\EndFor
\State \textbf{return} $\bigcup_{\ell=1}^n \hat{D}_\ell$
\EndFunction
\end{algorithmic}
\end{algorithm}

	\begin{proof}[Proof of \Cref{lem:rg-correctness}]
    In order to bound the probability that $\forall \ell$, $\hat{D}_\ell = D_\ell$,
    we consider the opposite and examine the first $\ell$ such that $\hat{D}_\ell \neq D_\ell$.
    Then we have the following
	\[
    \Pr\left[\exists 1\leq \ell\leq n,\hat{D}_\ell\neq D_\ell\right]\leq \sum_{\ell=0}^{n-1} \Pr\left[\hat{D}_{\ell+1}\neq D_{\ell+1}\mid \hat{D}_{1}=D_{1},\ldots ,\hat{D}_{\ell}=D_{\ell}\right].
    \]
	
	Fix some $\ell \in [1,n]$. 	
	In Skip-BFS, only vertices in $\{v_i^{(\ell)}\mid i\in \mathcal{I}^{(\ell)}\}$ are explored in iteration $\ell$. 
 Conditioning on $\hat{D}_{1}=D_{1},\ldots ,\hat{D}_{\ell}=D_{\ell}$, we have $\hat{D}_{\ell+1} = D_{\ell+1}$ if and only if for every omitted vertex $v_i^{(\ell)}$ ($i\in [W_\ell]\setminus \mathcal{I}^{(\ell)}$), every its predecessor $u\in \gIn_{\vec G}(v_i^{(\ell)})$ is explored by some $v_j^{(\ell)}$ ($j\in \mathcal{I}^{(\ell)}$), i.e. 
	$$\hat{D}_{\ell+1} = D_{\ell+1} \iff \forall i\in [W_\ell]\setminus \mathcal{I}^{(\ell)},u\in \gIn_{\vec G}(v_i^{(\ell)})\setminus D_{\leq \ell} ,\exists j\in \mathcal{I}^{(\ell)},(u,v_j^{(\ell)})\in \vec{E}.$$ 

    For every omitted vertex $v_i$ ($i\in [W_\ell]\setminus \mathcal{I}$), there are at least $R$ edges between $\mathcal{I}_{<i}:=\{v_j\mid j\in \mathcal{I},j<i\}$ and $v_i$. 
    We count the number of random edges (edges not in $\Gbase$) by $\hat{S}_i=\sum_{j\in \mathcal{I}_{<i},(v_i,v_j)\notin \Ebase} \vec{E}(v_i,v_j)$. We have $\hat{S}_i \geq R-k=18 \log N$ (with probability $1$). 

    In \Cref{alg:partition-rand-oracle}, all edges between $\{v_j\mid j\in \mathcal{I},j<i\}$ and $v_i$ remains unrevealed until we calculate $\hat{c}_i$. Hence $\hat{S}_i$ is sum of $S_i$ i.i.d. Bernouli variables, where $S_i = \sum_{j\in \mathcal{I}_{<i},(v_i,v_j)\notin \Ebase} \mathbb{I}(v_i<v_j)$. By Chernoff bound, for every $i\in [W_\ell], S_i \geq 18\log N$,
	$$\Pr\left[\hat{S}_i>1.1pS_i \mid v_1,v_2,\ldots ,v_{i-1} \right]\leq \exp(-1.21pS_i/2.1) \leq N^{-6}.$$

 \begin{comment}
	Define $\hat{T}_i:=\{v_j\mid (v_i,v_j)\in \vec{E},j<i,(v_i,v_j)\notin \vec{E}_{base}\}$. 
 By $i\notin \mathcal{I}^{(\ell)}$, we have $|\hat{T}_i|\geq R-k$.

	Define $T_i:=\{v_j \mid v_i<v_j,j<i,(v_i,v_j)\notin\vec{E}_{base}\}$. 
	
	According to our oracle model [how?], all edges between $\{v_j\mid j\in \mathcal{I},j<i\}$ and $v_i$ remains unrevealed until the calculation of $c_{i-1}(v_i)$. 
\end{comment}

For other $i$ such that $S_i < 18\log N$, we have $\hat{S}_i \leq S_i < 18\log N$, hence $i$ must be in $\mathcal{I}$. 

Define the bad event $\mathcal{E}_1$ as $\exists i\in [W_\ell], S_i \geq 18\log N, \hat{S}_i>1.1pS_i$. 
By taking union bound over every $i\in [W_\ell]$ such that $S_i \geq 18\log N$, 
for every permutation $\mathbf{v}=(v_1,v_2,\ldots ,v_{W_\ell})$,
$\Pr[\mathcal{E}_1\mid \mathbf{v}] \leq N^{-5}$. 

For every $i\in [W_\ell]\setminus \mathcal{I}$, $\hat{S}_i\leq 1.1pS_i$ implies $S_i \geq \frac{\hat{S}_i}{1.1p}\geq 10\log N/p$, which means there are at least $10\log N/p$ vertices $v_j$ ($j\in \mathcal{I}$) such that $v_i<v_j<\pvt$. Notice all edges between $\hat{D}_\ell$ and $V\setminus \hat{D}_{\leq \ell}$ remains unrevealed until the last step of level $\ell$. For every $u<v_i$, 
 $$\Pr[\exists u<v_i,\forall j\in \mathcal{I},(u,v_j)\notin \vec{E}\mid S_i\geq 10\log N/p] \leq n(1-p)^{10\log N/p} \leq N^{-9}.$$

 Define the bad event $\mathcal{E}_2$ as $\exists i\in [W_\ell]\setminus \mathcal{I}$, 
 such that $S_i \geq 10\log N/p$ and $\exists u<v_i,\forall j\in \mathcal{I},(u,v_j)\notin \vec{E}$. We have for every $\mathbf{v}=(v_1,v_2,\ldots ,v_{W_\ell})$, $\Pr[\mathcal{E}_2 \mid \mathbf{v}] \leq N^{-9}$.

 Let $\mathrm{EQ}_\ell$ be the event that $\hat{D}_{1}=D_{1},\ldots ,\hat{D}_{\ell}=D_{\ell}$. Let $\mathcal{E}$ be the event that $\exists i\in [W_\ell]\setminus \mathcal{I}^{(\ell)},u\in \gIn_{\vec G}(v_i^{(\ell)})\setminus D_{\leq \ell} ,\forall j\in \mathcal{I}^{(\ell)},(u,v_j^{(\ell)})\in \vec{E}$.
 We conclude that
 \begin{align*}
     \Pr\left[\exists 1\leq \ell\leq n,\hat{D}_\ell\neq D_\ell\right]&\leq \sum_{\ell=0}^{n-1} \Pr\left[\hat{D}_{\ell+1}\neq D_{\ell+1}\mid \mathrm{EQ}_\ell\right] \\
      &\leq \sum_{\ell=0}^{n-1}\sum_{\mathbf{v}}\Pr[\mathcal{E}\mid \mathbf{v},\mathrm{EQ}_\ell]\Pr[\mathbf{v}\mid \mathrm{EQ}_\ell] \\
     &\leq \sum_{\ell=0}^{n-1}\sum_{\mathbf{v}}(\Pr[\mathcal{E}_1 \mid \mathbf{v},\mathrm{EQ}_\ell]+\Pr[\mathcal{E}_2 \mid \mathbf{v},\mathrm{EQ}_\ell])\Pr[\mathbf{v}\mid \mathrm{EQ}_\ell]\\
     &\leq n^{-4}.
 \end{align*}
	
\begin{comment}
    Let event $\mathcal{E}$ be ...
 
	By union bound,
	$|\hat{T}_i|\leq 1.1p|T_i|$ holds for all every $i$ such that $|T_i|\geq R-k$. For every $i$, either $|\hat{T}_i|\leq |T_i|<R-k$ or $|T_i|\geq |\hat{T}_i|/(1.1p)$ holds, which implies $\forall i,|\hat{T}_i|\geq R-k,|T_i|\geq (R-k)/(1.1p)$. 
 [rewrite]
Then $|T_i|\geq |\hat{T}_i|/(1.1p)$, and then we have $\Pr[\exists v_j \in T_i, (u,v_j)\in \vec{E}]\geq 1 - (1-p)^{(c_R/1.1)\cdot p^{-1}\log n} \geq 1-n^{-c_R/2}$
\end{comment}

	\end{proof}
	
	\subsection{Efficiency}

    In this section, we show that \textsc{Skip-BFS} uses $O(nk_\prec \log^2 N)$ queries with high probability. 
	
	We first prove a property of random graphs under standard setting, which is independent of our algorithm. \Cref{lem:rg-short-dist} shows that most vertices are explored in early iterations $\ell = 1,2,\ldots ,6\log N$. 
	
	\begin{lemma}\label{lem:rg-short-dist}
		With probability at least $1-2N^{-4}$, for every $v\in V_\prec$, if there exists a path from $v$ to $\pvt$ in $\Gbase$ with length no less than $192 \log^2 N/p$, then $d_{\vec G}(v,\pvt) \leq 6 \log N$. 
	\end{lemma}
	
	\begin{proof}
		Fix some $v$ such that there exists a path from $v$ to $\pvt$ in $\Gbase$, denoted by $v\to u_1\to u_2\to \ldots \to u_{L_v} \to \mathrm{pivot}$ with $L_v \geq 192 \log^2 N/p$. We partition $u_1,u_2,\ldots ,u_{192\log^2 N/p}$ into $r=6\log N$ continuous segments of size $32 \log N/p$, denoted by sets $S_1,S_2,\ldots ,S_r$. 
		
		For each segment $S_i$, denote $\hat{S}_i$ as the vertices $s$ such that there exists a path $v\to s_1\to s_2 \to \ldots \to s_{i-1}\to s$ in $\vec G$ where $s_1\in S_1,s_2\in S_2,\ldots ,s_{i-1}\in S_{i-1}$. 
		We can also define $\hat{S}_i$ in a recursive way $\hat{S}_i=\{s\in S_i\mid \exists u\in \hat{S}_{i-1},(u,s)\in \vec{E}\}$. 
		In \Cref{lem:rg-sd-seg-inc}, we count the number of vertices in $\hat{S}_i$, and we use this as a lower bound of $|\{s\in S_i\mid d_{\vec G}(v,s)\leq i\}|$. 

\begin{lemma}\label{lem:rg-sd-seg-inc}
    Let event $\mathcal{E}$ be $\forall 1\leq i\leq r$, $|\hat{S}_i| \geq \min\{2^i, p^{-1}\}$. We have $\Pr[\mathcal{E}] \geq 1-N^{-5}$. Moreover, event $\mathcal{E}$ only depends on $\{\vec{E}(x,y)\mid x,y\in \{v,u_1,u_2,\ldots ,u_{L_v}\}\}$.  
\end{lemma}

\begin{proof}
Consider the probability that $|\hat{S}_j| \geq \min\{2^j, p^{-1}\}$ holds for $j=1,2,\ldots ,i-1$ and fail on $j=i$. 

Fix set $\hat{S}_{i-1}$ and let $s=|\hat{S}_{i-1}|$. For any vertex $u\in S_i$, we have
$$\Pr[u\in \hat{S}_i \mid \hat{S}_{i-1}]\geq 1 - (1-p)^{s} \geq 1-\exp(-ps)$$

\begin{itemize}
	\item When $ps< 1$, we have $1-\exp(-ps)\geq  ps-p^2s^2/2\geq ps/2$, set $q:=ps/2$. 
	\item Otherwise, we have $1-\exp(-ps)\geq 1-e^{-1}$, set $q:=1-e^{-1}$. 
\end{itemize}

For every vertex $s\in S_i$, define the indicator random variable $x_s=\mathbb{I}(s\in \hat{S}_i)$. Conditioning on $\hat{S}_{i-1}$, $x_s$ only depends on $\{E(u,s)\}_{u\in \hat{S}_{i-1}}$, which means the random variables $\{x_s\}_{s\in S_i}$ are mutually independent. 

If $q=ps/2$, by Chernoff bound, 
\begin{align*}
\Pr\left[\sum_{s\in S_i} x_s \leq 2^i \right] &\leq \Pr\left[\sum_{s\in S_i} x_s \leq \frac{1}{8}|S_i|q \right] \leq \exp\left(-\frac{49}{128} |S_i|q \right)<N^{-6}
\end{align*}

If $q=1-e^{-1}$, by Chernoff bound, 
\begin{align*}
\Pr\left[\sum_{s\in S_i} x_s \leq p^{-1} \right] &\leq \Pr\left[\sum_{s\in S_i} x_s \leq \frac{1}{8}|S_i|q \right] \leq \exp\left(-\frac{49}{128} |S_i|q \right)<N^{-6}
\end{align*}

Apply union bound over $i=1,2,\ldots ,r$, the lemma statement fails with probability at most $N^{-5}$. 
This finishes the proof of \Cref{lem:rg-sd-seg-inc}.
    \end{proof}

 Suppose $\mathcal{E}$ is the event defined in \Cref{lem:rg-sd-seg-inc}.
 When the event $\mathcal{E}$ happens, we have $\sum_{\log (1/p)}^{6\log N-1} |\hat{S}_i|\geq 5\log N/p$. That is to say, there are at least $5\log N/p$ vertices $u$ in $u_1,u_2,\ldots ,u_{L_v}$ such that $d(v,u)<6\log N$. By Lemma \ref{lem:rg-sd-seg-inc}, we have
	\begin{align*}
		\Pr[d(v,\pvt) > 6\log N]&\leq \Pr[\neg \mathcal{E}] + \Pr[\forall i<6\log N,s\in \hat{S}_i,(s,\pvt)\notin \vec{E}\mid \mathcal{E}] \\
		&\leq N^{-5}+(1-p)^{5\log N/p} \\
		&\leq 2N^{-5}
	\end{align*}
 
	The proof is completed by taking a union bound on all vertices. 
\end{proof}	

For every vertex $v$ not skipped in Skip-BFS, exploring $v$ costs $|\gAdj_{G}(v)|$ queries, which is $O(np)$ in expectation. By directly applying Chernoff bound and union bound, we show that the exploration cost of $s$ vertices is $\tilde{O}(s\cdot np)$ with high probability in \Cref{lem:count-adj}. 

\begin{restatable}{lemma}{CountAdj}\label{lem:count-adj}
    For every $S\subseteq V$, with probability of at least $1-N^{-6}$, $\sum_{v\in S}|\gAdj_G(v)| \leq \max\{4np|S|,8\log N\}$. 
\end{restatable}

\begin{proof}
    The proof can be found in \Cref{sec:proof-count-adj}. 
\end{proof}

In the following discussion, we assume that $n > 8\log N$, otherwise we can simply query all the edges. 

For every $a,b\in V,a\prec b$, define $\lp(a,b)$ as the length of the longest path from $a$ to $b$ in $\Gbase$. By \Cref{lem:rg-short-dist}, we can separate the vertices into two parts. 
\begin{itemize}
    \item Vertices $v$ such that $\lp(v,\pvt) \geq 192\log^2 N/p$ : All of them are contained in $\hat{D}_0,\hat{D}_1,\ldots \hat{D}_{6\log N}$. Their exploration cost is discussed in \Cref{lem:rg-single-round}. 
    \item Vertices $v$ such that $\lp(v,\pvt) < 192\log^2 N/p$ : We suppose that all of them are explored. Their exploration cost is bounded by \Cref{lem:late-iter}. 
\end{itemize}

\begin{lemma}\label{lem:late-iter}
    With probability at least $1-N^{-6}$, we have
    $$\sum_{v:\lp(v,\pvt) < 192\log^2 N/p)} |\gAdj_{G}(v)| = O(nk_\prec\log^2 N).$$ 
\end{lemma}

\begin{proof}

    Let $(C_1,C_2,\ldots ,C_{k_\prec})$ be a chain decomposition of $(V_\prec,\prec)$. 
    For each chain $C_i$, let $C_i = c_{i,1} \prec c_{i,2} \prec \ldots \prec c_{i,s_i}$. 
    For every $1\leq i\leq k_\prec, 1\leq j\leq s_i$, 
    $\lp(c_{i,j},c_{i,s_i}) \geq s_i - j$ holds because there exists a path from $c_{i,j}$ to $c_{i, s_i}$ in $\Gbase$ that contains vertices $c_{i,j},c_{i,j+1},\ldots ,c_{i,s_i}$. Then we have $\lp(c_{i,j}, \pvt) \geq \lp(c_{i,j}, c_{i, s_i}) + \lp(s_i, \pvt) \geq s_i - j$. 

    That is to say, only vertices in $\{c_{i,j} \mid 1\leq i\leq k_\prec, s_i - 192\log^2 N/p < j\leq s_i\}$ may satisfy $\lp(v,\pvt) < 192\log^2 N/p$. 
    Define $X=\{v\in V_\prec \mid \lp(v,\pvt) < 192\log^2 N/p\}$, we have $|X| \leq 192k_\prec\log^2 N/p$. 

    By \Cref{lem:count-adj}, with probability at least $1-N^{-6}$, $\sum_{x\in |X|}|\gAdj_G(v)| \leq 4np|X|\leq 800nk_\prec\log^2 N$. 
\end{proof}

Now it is time to analyze the exploration cost in iteration $\ell =1,2,\ldots ,6\log N$. In order to prove \Cref{lem:rg-single-round}, we give another equivalent description of \Cref{alg:partition-rand} under the oracle model in \Cref{alg:partition-rand-oracle2}. \Cref{alg:partition-rand-oracle2} not only produces the same output as \Cref{alg:partition-rand}, but also use the exactly same number of queries. Unlike \Cref{alg:partition-rand-oracle}, \Cref{alg:partition-rand-oracle2} is basically a formalized version of \Cref{alg:partition-rand}, except two small modifications. 
\begin{itemize}
    \item The evaluation of $\hat{D}_{\ell+1}$ is delayed to the end of iteration $\ell$. 
    \item Instead of randomly permute all vertices in $\hat{D}_\ell$ and iteratively check whether they are alive ($c[v]>0$), we directly select alive vertices (uniformly at random). 
\end{itemize}

Similar to \Cref{alg:partition-rand-oracle}, we use ``$\ell$-th iteration'' or iteration $\ell = x$ to refer to the for-loop in line \ref{line:oracle2-for1}, use ``$i$-th iteration'' or iteration $i = x$ to refer to the while-loop in line \ref{line:oracle2-for2}. 

	\begin{algorithm}[h]
\caption{SkipBFS in Oracle Model}\label{alg:partition-rand-oracle2}
\begin{algorithmic}[1]
\Function{Oracle-Skip-BFS2}{$\pvt$}
\State $\hat{D}_0\gets \{\pvt\}$
\State $\hat{D}_1\gets \{\mathrm{in}(\mathrm{pivot})\}$
\State let parameter $R\gets k+18\log N$
\For{$\ell\gets 1,2,\ldots ,n$ \label{line:oracle2-for1}}  
    \State for every $v\in \hat{X}_\ell$, set $c_0^{(\ell)}(v)\gets R$
    \State $X_0^{(\ell)}\gets \hat{D}_\ell$
    \State $i\gets 1$
    \While{$X_{i-1}^{(\ell)} \neq \varnothing$ \label{line:oracle2-for2}}
        \State pick vertex $x_i^{(\ell)}$ from $X_{i-1}^{(\ell)}$ uniformly at random
\State reveal and query all edges between $\{x\in X_{i-1}^{(\ell)} \mid x \prec x_i^{(\ell)}\}$ and $x_i^{(\ell)}$ \label{line:oracle2-query1}
        \State for every $v\in \hat{D}_\ell$, set $c_i^{(\ell)}(v)\gets R - \sum_{j=1}^i \vec{E}(v,x_j^{(\ell)})$
        \State $X_i^{(\ell)}\gets \{v \in \hat{D}_\ell \mid c_i^{(\ell)}(v) > 0\}$
        \State $i\gets i+1$
    \EndWhile
    \State $M_\ell \gets i-1$ \Comment{$M_\ell$ denotes the number of iterations}
    \For{$j=1,2,\ldots ,M_\ell$}
\State reveal and query all unrevealed edges adjacent to $x_j^{(\ell)}$ \label{line:oracle2-query2}
\State add all vertices $u\in V\setminus \hat{D}_{\leq \ell}$ such that $(u,x_j^{(\ell)})\in \vec{E}$ to $\hat{D}_{\ell+1}$
	\EndFor
\EndFor
\State \textbf{return} $\bigcup_{\ell=1}^n \hat{D}_\ell$
\EndFunction
\end{algorithmic}
\end{algorithm}

\begin{lemma}\label{lem:rg-single-round}
    With probability at least $1-2N^{-5}$, 
    $\sum_{\ell = 0}^{6\log N}\sum_{i=1}^{M_\ell} |\gAdj_{G}(x_i^{(\ell)})| = O(nRk\log^2 N)$. 
\end{lemma}

\begin{proof}
    Fix some $0\leq \ell \leq n$, in the following discussion, the superscript $(\ell)$ in the variables of \Cref{alg:partition-rand-oracle} is omitted. 

	Pick a chain decomposition $\mathcal{C}=(C_1,C_2,\ldots ,C_{k_\prec})$ of $(V_\prec,\prec)$. 
    We analyze the exploration cost of each chain individually. 
    Let $C\in \mathcal{C}$ be the chain to be analyzed. 
    Define $\hat{X}_i = X_i \cap C$.

    Define $S_i := \sum_{v\in \hat{X}_i} c_i(v)$ as the sum of counters after the $i$-th iteration. 
    We show that each iteration $i$ with $x_i \in C$ typically reduces $S_i$ by $\Theta(p|\hat{X}_{i-1}|)$ in \Cref{lem:chain-reduce}.

\begin{lemma}\label{lem:chain-reduce}
    For every $1\leq i\leq M_\ell$ and every $u_1,u_2,\ldots ,u_{i-1}$ such that the size of $\hat{X}_{i-1}$ (uniquely determined by $x_1=u_1,x_2=u_2,\ldots ,x_{i-1}=u_{i-1}$) is no less than $12p^{-1}$, we have
    $$\Pr\left[S_{i-1}-S_i \geq p|\hat{X}_{i-1}|/6 \mid x_i\in C,x_1=u_1,x_2=u_2,\ldots ,x_{i-1}=u_{i-1}\right] \geq 2/5.$$
\end{lemma}

\begin{proof}
    Fix some $i$ and $u_1,u_2,\ldots ,u_{i-1}$. 
    Let event $\mathcal{E}$ be $x_i\in C, x_1=u_1,x_2=u_2,\ldots ,x_{i-1}=u_{i-1}$. 
    
    Conditioning on $\mathcal{E}$, $x_i$ is uniformly distributed over $\hat{X}_{i-1}$. Then with probability of at least $1/2$, we have $\rnk_{\hat{X}_{i-1}}(x_i) \geq \lfloor |\hat{X}_{i-1}|/2 \rfloor$. 

    In \Cref{alg:partition-rand-oracle2}, none of the events in $\{E(u,v)\mid u,v\in X_{i-1}\}$ is revealed before the $i$-th iteration. This implies $\{E(u,v)\mid u,v\in X_{i-1}\}$ is independent of $\mathcal{E}$. 
Let $\mu=p(\rnk_{\hat{X}_{i-1}}(x_i)-1)$. 
By Chernoff bound, we have 
\begin{align*}
    \Pr\left[|\{u\in \hat{X}_{i-1}\mid (u,x_i)\in \vec{E}\}| \leq \mu/3 \mid \mathcal{E}\right] \leq \exp(-2\mu/9) \leq 1/10. 
\end{align*}

By union bound,
\begin{align*}
    \Pr\left[S_{i-1}-S_i \geq p|\hat{X}_{i-1}|/6 \mid \mathcal{E}\right] &\geq 
    \Pr\left[\rnk_{\hat{X}_{i-1}}(x_i) \geq \lfloor |\hat{X}_{i-1}|/2 \rfloor,|\{u\in \hat{X}_{i-1}\mid (u,x_i)\in \vec{E}\}| > \mu/3 \mid \mathcal{E}\right] \\
    &\geq 2/5. 
\end{align*}
This concludes our proof. 
\end{proof}

Let $I_1,I_2,\ldots ,I_s$ be indices $i$ such that $x_i\in C$ and $|\hat{X}_{i-1}| \geq 12p^{-1}$, sorted in ascending order. Let $s^* = 50R\log N/p$ be an upper bound of $s$. We show that $s\leq s^*$ with high probability. 

Define $s^*$ indicator variables $Y_1,Y_2,\ldots ,Y_{s^*}$ as below
\begin{align*}
    Y_i := \begin{cases}
        1 & \text{$i>s$ or $|\hat{X}_{I_i-1}| < 12p^{-1}$} \\
        \mathbb{I}(S_{I_i-1}-S_{I_i} \geq p|\hat{X}_{I_i-1}|/6) & \text{otherwise}
    \end{cases}
\end{align*}

For every $1\leq i\leq M_\ell$, $S_i \leq R|\hat{X}_i|$. Hence $S_{i-1} - S_i \geq p|\hat{X}_{i-1}|/6$ implies $S_i \leq (1-\frac{p}{6R})S_{i-1}$. The sum of counters will be reduced to $0$ after being reduced by a factor of $p/6R$ for at most $\log_{1-\frac{p}{6R}}(nR)$ times. Let $L = 10R\log n/p \geq \log_{1-\frac{p}{6R}}(nR)$, $Y_1+Y_2+ \ldots +Y_{s^*} \geq L$ implies $s\leq s^*$. 

As $Y_1,Y_2,\ldots ,Y_{s^*}$ are not independent, we cannot apply Chernoff bound directly. Here we use \Cref{lem:chernoff-non-independent} instead. 

\begin{restatable}{lemma}{ChernoffVariant}\label{lem:chernoff-non-independent}
    For 0/1 random variables $Y_1,Y_2,\ldots ,Y_N$, if there exists $\mu$ such that for every $1\leq i\leq N$, $\Pr[Y_i=1\mid Y_1,Y_2,\ldots ,Y_{i-1}] \geq \mu$ holds for every $Y_1,Y_2,\ldots ,Y_{i-1}$. Then we have, for every $0<\delta<1$,
    $$\Pr[Y_1+Y_2+\ldots +Y_N < (1-\delta)N\mu] \leq \exp(-N\delta^2\mu/2).$$
\end{restatable}
\begin{proof}
    The proof can be found in \Cref{sec:proof_lem:chernoff-non-independent}.
\end{proof}

By \Cref{lem:chain-reduce}, for every $1\leq i\leq s^*$, $\Pr[Y_i=1\mid Y_1,Y_2,\ldots ,Y_{i-1}] \geq 0.4$ for every $Y_1,Y_2,\ldots ,Y_{i-1}$. 
Then use \Cref{lem:chernoff-non-independent}, we have $\Pr[s \leq s^*] \geq \Pr[Y_1+Y_2+\ldots + Y_{s^*} \geq L] \geq 1 - N^{-10}$. 

Recall that $I_s$ is the largest index $i$ such that $x_i\in C,|\hat{X}_{i-1}| \geq 12p^{-1}$, this implies $|\{x_1,x_2,\ldots ,x_{M_\ell} \}\cup C| \leq s + 12p^{-1}$. 
Taking a union bound over $C\in \mathcal{C}$. We have $\Pr[M_\ell \leq 62k_\prec R\log N/p] \geq 1-N^{-9}$. 

Now we are ready to count the total number of queries in iteration $\ell$, i.e. $\sum_{i=1}^{M_\ell} |\gAdj_G(x_i^{(\ell)})|$. 

In \Cref{alg:partition-rand-oracle2}, the queries only happens at line \ref{line:oracle2-query1} and line \ref{line:oracle2-query2}. 
\begin{itemize}
    \item Line \ref{line:oracle2-query1}: Although these queries are not independent of $x_1,x_2,\ldots x_{M_\ell}$, which means we cannot directly apply Chernoff bound, we notice that when an edge $(x,x_i)\in \vec{E}$ is revealed ($x\in X_{i-1}$), the counter of $x$ is reduced by 1, i.e. $c_{i-1}(x) - c_i(x) = 1$. This implies the total number of queries happens in line \ref{line:oracle2-query1} is no more than the sum of all counters, i.e. $nR$. 
    \item Line \ref{line:oracle2-query2}: These edges are independent of $x_1,x_2,\ldots ,x_{M_\ell}$. To be precisely, let $F(y_1,y_2,\ldots ,y_M)$ denotes the set of vertex pairs that are revealed at line 11 when $M_\ell = M, x_1=y_1,x_2=y_2,\ldots ,x_M=y_M$. 
    For every $1\leq M\leq |\hat{D}_\ell|$ and $y_1,y_2,\ldots ,y_M \in V_\prec$, the events $\{E(u,y_i) \mid u\in V, (u,y_i)\notin F(y_1,y_2,\ldots ,y_M) \}$$ \}$ are independent of the event $M_\ell = M, x_1=y_1,x_2=y_2,\ldots ,x_M=y_M$. By slightly modify the proof of \Cref{lem:count-adj}, we can show that with probability at least $1-N^{-6}$, the total number of queries happens in line \ref{line:oracle2-query2} is no more than $4npM_\ell$. 
\end{itemize}

By taking a union bound with $\Pr[M_\ell > 62k_\prec R\log N/p] < N^{-9}$, we show that with probability at least $1-2N^{-6}$, $\sum_{i=1}^{M_\ell} |\gAdj_G(x_i^{(\ell)})| \leq 250nk_\prec R\log N$. 

We conclude our proof by taking one more union bound over $\ell = 0,1,\ldots ,6\log N$. 

\end{proof}

\begin{proof}[Proof of \Cref{thm:er}]
    By \Cref{lem:rg-correctness}, we may assume for every $\ell$, $\hat{D}_\ell = D_\ell$, which means \textsc{Skip-BFS} outputs $V_\prec$ correctly. 
    
    Combining \Cref{lem:rg-correctness}, \Cref{lem:rg-short-dist} and \Cref{lem:rg-single-round}, we may assume all vertices with $\lp(v,\pvt) \geq 192\log^2 N/p$ are visited by \text{Skip-BFS} in iteration $\ell = 0,1,\ldots ,6\log N$, furthermore, these iterations uses $O(nRk_\prec\log^2 N)$ queries. 

    By \Cref{lem:late-iter}, we may assume exploring all vertices with $\lp(v,\pvt) < 192\log^2 N/p$ uses $O(nRk_\prec\log^2 N)$ queries. 

    By taking union bound on the probability of these assumptions, we conclude that, with probability at least $1-10N^{-4}$, \textsc{Skip-BFS} outputs $V_\prec$ in $\tilde{O}(nk_\prec \cdot k)$ queries. $V_\succ$ can also be computed by \textsc{Skip-BFS} in $\tilde{O}(nk_\succ \cdot k)$ queries, $V_{\incom}$ is computed by $V_\incom = V \setminus (V_\prec \cup V_\succ)$. 

    By applying \Cref{lem:part2LE-main} with $\varepsilon = 10N^{-4}$ and $f(n,k) = k$, we conclude that the combination of \Cref{alg:LEtoPoset}, \Cref{alg:part-to-LE} and \Cref{alg:partition-rand} solves GPS problem on \ER query graphs using $\tilde{O}(nk^2)$ queries. 
\end{proof}

\section{Partition Algorithms for Complete Bipartite Query Graphs}

\thmbipartite*
We prove \Cref{thm:bipartite} in this section.
Recall that our input is an $n$-vertex undirected query graph $G=(V,E)$
which is complete bipartite, and a pivot $\pivot$.
There is an underlying directed graph $\vec{G}=(V,\vec{E})$, and an underlying poset $\mathcal{P}=\mathcal{P}(\vec{G})$.
We aim to solve $V_{\succ \pivot}$, $V_{\prec \pivot}$, and $V_{\nsim\pivot}$ with respect to $\mathcal{P}$.
Let $k_{\succ \pivot} = k_{V_{\succ \pivot}}, k_{\prec \pivot} = k_{V_{\prec \pivot}}$ (recall that $k_X$ denotes the width of $\mathcal{P}(\vec{G}[X])$). 
The proof of \Cref{thm:bipartite} relies on the following main lemma which we  prove in this section.

\begin{lemma}
    \label{lem:bipartitemain}
    When $G$ is a complete bipartite graph, we can detect $V_{\succ \pivot}$, $V_{\prec \pivot}$, and $V_{\nsim \pivot}$ with $O(( k_{\succ \pivot} + k_{\prec \pivot}) \cdot n \log n)$ queries and use $O(( k_{\succ \pivot} + k_{\prec \pivot}) \cdot n \log n)$ time, in expectation.
\end{lemma}

\begin{proof}[Proof of \Cref{thm:bipartite}]

Notice that \Cref{lem:bipartitemain} satisfies the condition of \Cref{lem:part2LE-main2}.
To see this, when we take $G[X]$ as an input, it is still complete bipartite, and $V_{\prec \pivot}$, $V_{\nsim \pivot}$ and $V_{\succ \pivot}$
is exactly $\Lset{X}{p},\Mset{X}{p},\Rset{X}{p}$.
Our query complexity can be viewed as $O(|X|(k_{X_{\prec \pivot}} + k_{X_{\succ \pivot}})\log|X|)$,
where the $\log |X|$ term may be bounded by $f(n,k)=\log n$.
Thus, we can apply \Cref{lem:bipartitemain} to obtain a linear extension using $\tilde{O}(nk)$ queries.
Finally, we finish the proof of \Cref{thm:bipartite} by combining with \Cref{thm:linearexttoposet}.
\end{proof}

\begin{remark}
    \label{remark:bipartite}
    We notice that one can reduce the problem of GPS on a complete query graph $G = (V, E)$ to one with a complete bipartite graph which is only constant times larger than $G$.
    For every $v \in V$, create two vertices $v_L$ and $v_R$, and define $v_L \prec v_R$.
    In the bipartite graph, the vertex set is $L \cup R$ where $L := \{ v_L : v \in V \}$ and $R := \{ v_R : v \in V \}$.
    Then for every edge $(u, v) \in E$,
    if and only if $u \prec v$, define $u_L \prec v_R$ and $u_R \prec v_L$.
    A poset sorting on this complete bipartite case uniquely maps back to a poset sorting in the original complete graph case, and solves it.
\end{remark}

\paragraph{Proof overview of \Cref{lem:bipartitemain}.}

Without loss of generality, we assume $\pivot$ is in $B$, and we only present how to detect $V_{\succ \pivot}$ with $O(k_{\succ \pivot}n \log n)$ queries. $V_{\prec \pivot}$ can be detected by a symmetric process, and so we will also have $v_{\nsim \pivot}$. Hereafter, we use $k=k_{\succ \pivot}$ for simplification. It is easy to determine the role of every vertex in $A$ by $O(n)$ queries because the graph is complete bipartite. However, the challenge is the $B$ side. 

It can be done by determining the set of \emph{minimal} vertices in $A_{\succ \pivot}$. (We say a vertex $v$ is \emph{minimal} in a set $S$ if for all the other vertices $v'\in S$, we have $v'\succ v$ or $v'\nsim v$.) However, it is tricky because we can not know the ``real" minimal vertices unless $B_{\succ \pivot}$ is given. We will first introduce a subroutine called \findmin that can output(randomly) a ``local" minimal vertex based only on a subset of $V_{\succ \pivot}$. Then, we use the ``local" minimal vertex to expand $B_{\succ \pivot}$ and call \findmin again. We prove that $B_{\succ \pivot}$ can be completely recovered after certain times of iterations.  

\subsection{Finding A (Local) Minimal Vertex}
\label{section:findmin}
First, we introduce \findmin that can output a ``local" minimal vertex. Given two vertex sets $A'\subseteq A$ and $B'\subseteq B$, the subroutine works on the induced subgraph $G[A'\cup B']$. The subroutine first picks a vertex $a$ from $A'$ as the starting point. Then, we keep selecting new vertices uniformly at random, moving to it if the new vertex is smaller and deleting it otherwise. Finally, we will stop at a minimal point.   
\begin{algorithm}[H]     
\caption{Randomly Finding A Minimal Vertex}
    \begin{algorithmic}[1]
        \Function{\findmin}{$A',B'$} \Comment{we work on $G'=G[A' \cup B']$}
            \State pick a starting vertex $\vmin$ in $A'$ uniformly at random 
            \State delete $\vmin$ from $A'$ 
            \While{$A' \neq \emptyset$ and $B' \neq \emptyset$}
                \If{$\vmin\in A'$}
                    \State uniformly pick a random vertex $u$ in $B'$
                    \State $B'\leftarrow B' \setminus \{u\}$
\Else
                    \State uniformly pick a random vertex $u$ in $A'$
                    \State $A'\leftarrow A' \setminus \{u\}$
                \EndIf
            \If{$u\prec \vmin$} \Comment{by the result of querying the edge $(u,\vmin)$}
                \State $\vmin\leftarrow u$
            \EndIf
            \EndWhile
            \State \textbf{return} $\vmin$.
        \EndFunction
    \end{algorithmic}
\end{algorithm}  

{

Since the old $\vmin_i$ is eliminated only if some vertex $u$ (which will be $\vmin_{i+1}$) is comparable and smaller than it, we regard $\vmin_{0}\succ \vmin_{1}...\succ \vmin_{N}=\vmin$ as a smaller chain, where $N$ is the length of the smaller chain. Then we discuss properties of \findmin.

\begin{lemma}[Cost]
    \findmin costs $O(n)$ queries and runs in $O(n)$. 
\end{lemma}
\begin{proof}
Because each round a vertex is deleted from $A'\cup B'$, \findmin terminates in at most $O(|A'| +|B'| )=O(n)$ rounds. We spend one query on edge $(u,\vmin)$ in each round. Besides, the other operations also run in $O(1)$ in each round. 
    We conclude that \findmin costs $O(n)$ queries and runs in $O(n)$. 
\end{proof}

\begin{lemma}[Minimal]
\label{lem:findmin2} 
The returned vertex $\vmin$ is minimal, i.e., 
\begin{itemize}
    \item If the $\vmin \in A'$, $\forall b\in B'$, we have $b\nsim \vmin$ or $b\succ \vmin$.
    \item If the $\vmin \in B'$, $\forall a\in B'$, we have $a\nsim \vmin$ or $a\succ \vmin$.
\end{itemize}

\end{lemma}
\begin{proof}
We prove case $\vmin\in A'$ as an example.
    When the algorithm returns $\vmin\in A'$, all $B'$ vertices have been compared to some $A'$ vertex on the smaller chain and out (incomparable or larger). Therefore, all the $B'$ vertices are incomparable or larger than $\vmin$. The proof for $\vmin\in B'$ is symmetric.   
\end{proof}
Let us discuss the property more. When we input the original graph $G$ to \findmin, \findmin can actually output a ``real" minimal vertex. However, when the input is only an induced subgraph, the output of \findmin may only be a ``local" minimal vertex. Consider if we select $a_1$ and $a_2$ included in $A'$, and the real relation between them is $a_1 \prec  a_2$, we may not know it only by querying $G[A'\cup B']$ because we miss some vertices in $B'$. We can only promise the minimal property based on the edges in the induced subgraph, as \Cref{lem:findmin2} claims. 

\begin{lemma}[$A$-first]
\label{lem:findmin2fo}
    If $v^* \in B'$, there must exist a vertex $a\in A'$, such that $a\succ v^*$.  
\end{lemma}
\begin{proof}
Because we always choose $a\in A'$ as the starting vertex, if $\vmin \in B'$, we must have $a \succ  \vmin$.
\end{proof}

{We don't have similar lemma for $\vmin\in A'$, because $\vmin$ may equal to $\vmin_{0}$.}

Finally, we discuss the randomness of the process. We observe that a ``real" smaller vertex should have more chance to be returned by \findmin even if we only work on an induced subgraph $G'$. 
\begin{lemma} 
\label{lem:findmin3} For $a_1,a_2\in A'$, if $a_1\prec a_2$ , $\p[\vmin = a_1] \geq \p[\vmin = a_2]$.
\end{lemma}
\begin{proof}

By \Cref{lem:findmin2}, if $a_2\in A'$ is larger than some $b\in B'$, $\p[\vmin = a_1]\geq\p[\vmin = a_2]=0$. The lemma holds. Then, we discuss the situation when $a_1,a_2$ are smaller than $\forall b\in B'$.

We write $\p[\vmin=a]$ in the following form according to the length $N$ of the smaller chain $\vmin_{0}\succ \vmin_{1}...\succ \vmin_{N}=\vmin$.
\begin{equation*}
\begin{aligned}
\p[\vmin=a]&=\p[\vmin=a,N=0]+\p[\vmin=a,N> 0]\\
&=\p[\vmin=a\mid N=0]\cdot\p[N=0]+\sum_{b\in B'_{\succ a}}\sum_{N}\p[\vmin=a\mid \vmin_{N-1}=b]\cdot\p[\vmin_{N-1}=b].
\end{aligned}
\end{equation*}

When $N=0$, because we choose $\vmin_0$ from $A'$ uniformly at random, $$\p[\vmin=a_1\mid N=0]=\p[\vmin=a_2\mid N=0].$$

When $N> 0$, because the next vertex $\vmin_{N}$ after $\vmin_{N-1}=b$ on the smaller chain is chosen uniformly at random in $A'_{\prec b}$, $\p[\vmin_{N}\notin A'_{\prec b}]=0$ while $\p[\vmin_{N}\in A'_{\prec b}]=1/|A'_{\prec b}|$. Then, we have: 
\begin{itemize}
    \item For $\vmin_{N-1}=b\in B'_{\succ a_2}\subseteq B'_{\succ a_1}$, because $a_1$ and $a_2$ are both in $A'_{\prec b}$, $$\p[\vmin=a_1\mid \vmin_{N-1}=b]=\p[\vmin=a_2\mid \vmin_{N-1}=b].$$
    \item For $\vmin_{N-1}=b\in B'_{\succ a_1}\setminus B'_{\succ a_2}$, because only $a_1$ is in $A'_{\prec b}$, $$\p[\vmin=a_1\mid \vmin_{N-1}=b]> \p[\vmin=a_2\mid \vmin_{N-1}=b]=0.$$
\end{itemize}
Therefore $\p[\vmin=a_1]\geq \p[\vmin=a_2]$.

\end{proof}
}

\subsection{Determining $B_{\succ \pivot}$}

{
Then, we present in detail how we construct $B_{\succ \pivot}$ iteratively by keep calling \findmin and appending vertices into $B_{\succ \pivot}$. In the first step, we use \findmin on $G[A_{\succ \pivot}]$. Because no vertices in $B$ are included, by the property of \findmin, we will get a vertex $a$ in $A_{\succ \pivot}$ uniformly at random. Then, we have that $B_{\succ a}$ must be a subset of $B_{\succ \pivot}$, and we union it into $\bunion$ (the currently discovered subset of $B_{\succ \pivot}$). After that, we move to the next round and call \findmin again on $G[A_{\succ \pivot} \setminus \{a\} \cup \bunion]$. Because we include all $B$-side vertices larger than $a$ and delete $a$ from the induced subgraph, \findmin will not return a vertex larger than $a$ (including $a$) again. We will find a ``smaller" minimal vertex in $A$ and use it to expand $\bunion$ again. 

Intuitively, let us focus on one chain in the chain decomposition of the poset. By the property of \findmin, smaller vertices should have more chance to be selected. Therefore, we can move across at least half of the vertices on this chain with at least half probability. As a result, if we repeat the process for $O(k\log n)$ rounds (we pay $O(n)$ cost in each round), we have visited all ``real'' minimal points on all the chains (at most $k$). Ideally, we are done in this state because all vertices in $B_{\succ \pivot}$ are included in $\bunion$. 

The second problem is how to figure out whether we are in such a good state so that we can terminate the iteration. Let us focus on one chain again. We observe that \findmin returns a $B$-side vertex if and only if all $A$-side vertices smaller than $b$ have been returned before and have been deleted, which means that this chain is completed. So, next, we propose to delete all the $A$-side vertices larger than $b$. Because \findmin is an $A$-first (by \Cref{lem:findmin2fo}) process, every $B$-side vertices larger than $b$ (including $b$) will not be visited again. Finally, the algorithm will terminate when all $A$-side vertices are deleted, and we prove the extra iterations where we get a $B$-side vertex can also be bounded in $O(k\log n)$. 
The formal description of this subroutine is presented in \Cref{alg:findlarge}.
}

\begin{algorithm}[H]
    \caption{Recover $V_{\succ \pivot}$}
    \label{alg:findlarge}
    \begin{algorithmic}[1]
        \Function{\findkmin}{G, $\pivot$}
\State $\Tilde{A} \leftarrow A_{\succ \pivot} $  \Comment{get $A_{\succ \pivot}$ by at most $n$ queries}
            \State $\bunion \gets \{\}$ \Comment{used to maintain the current detected subset of $B_{\succ \pivot}$}
            \While{$\Tilde{A}\neq\emptyset$}
            	\State $\vmin \gets \findmin(\Tilde{A}, \bunion)$ \Comment{$O(n)$ queries by \findmin}
                \If{$\vmin\in A$}
                    \State $\Tilde{A} \gets \Tilde{A} \setminus \{\vmin\}$
                    \State $\bunion\leftarrow \bunion\cup B_{\succ \vmin}$
                \EndIf
                \If{$\vmin\in B$}
                    \State delete $A_{\succ \vmin}$ from $\Tilde{A}$\Comment{get $A_{\succ \vmin}$ by at most $n$ queries}
                \EndIf
            \EndWhile
            \State \textbf{return} $A_{\succ \pivot} \cup \bunion$
        \EndFunction
    \end{algorithmic}
\end{algorithm}  

We first show that our algorithm will correctly recover $B_{\succ \pivot}$ when it terminates. 
\begin{lemma}
$\bunion= B_{\succ \pivot}$ when $\Tilde{A}=\emptyset$.
\end{lemma}
\begin{proof}
    Obviously, $\bunion\subseteq B_{\succ \pivot}$ because we never include any $b\prec \pivot$ into $\bunion$.
    We prove $B_{\succ \pivot}\subseteq \bunion$ by contradiction.
    Assumed that $\exists b\in B_{\succ \pivot}$, such that $b\notin\bunion$. It means that \findmin never returns any vertices $a$ in $A$ such that $a\prec b$. Notice that these vertices cannot be deleted by other $b' \in \bunion$, because otherwise, $b$ will also be included in $\bunion$ in the round we include $b'$. Thus, these $a\prec b$ will never be deleted unless they are returned by \findmin. It concludes the lemma. 
\end{proof}

Then, we move to the complexity. 
The complexity of the algorithm relies on the number of iterations it makes. The intuition is that we can promote each chain quickly. Let us focus on a chain decomposition $\mathcal{C}=\{C_1,C_2,\dots\}$. For each chain, we recover it to a real path on the complete bipartite graph starting from $\pivot$. Then, $C_i$ becomes $\{\pivot = b_{i0}\prec a_{i0}\prec b_{i1}\prec a_{i1}\prec b_{i2}\prec \dots\}$. Remark that it may let a vertex appear on multiple chains, but it does not matter in our analysis. For a vertex $a_{ij}$, we use $r_i(a_{ij})=j$ to denote its \emph{rank} on $i$. We also define $r_i(b_{ij}) = j$ symmetrically. 
For each chain, we use $\pi_i$ to denote its current state:
$
    \pi_i = \max\{ j \mid b_{ij} \not\in \bunion \}.
$
For each iterative round in our algorithm, we define two kinds of events as follows
\begin{itemize}
	\item $\mathcal{E}_i^A$: the \findmin function returns a vertex $v$ in $A$ and on $C_i$, and we call $\mathcal{E}_i^A$ happens on $v$.
	\item $\mathcal{E}_i^B$: the \findmin function returns a vertex $v$ in $B$ and on $C_i$, and we call $\mathcal{E}_i^B$ happens on $v$.
 
\end{itemize}
We remark that one call of \findmin can trigger more than one event because one vertex can appear on multiple chains. 
Consider an event $\finda_i$ that happens on vertex $v$ and on chain $C_i$ where $C_i$'s state is currently $\pi_i$. We call $\finda_i$ \emph{good} if the rank of $v$ is at most $\pi_i/2$. The following lemma shows why we promote each chain efficiently. 
\begin{lemma}
\label{lem:findmin3co}
The probability that $\finda_i$ is good is at least $1/2$. 
\end{lemma}
\begin{proof}
Consider $\finda_i$ happens on chain $C_i = \{ \pivot=b_{i0}\prec a_{i0}\prec b_{i1}\prec ...\}$ with state $\pi_i$. Assume $v^*$ is the returned vertex in $\tilde{A}$. By \Cref{lem:findmin2}, we have $r_i(\vmin) \leq \pi_i$.

Then we claim that vertices in $\{a\mid r_i(a)\leq\pi_i\}$ must haven't been deleted from $\Tilde{A}$, or $\pi_i$ has decreased to that rank. Therefore vertices in $\{a\mid r_i(a)\leq\pi_i\}$ are all possible to be the returned vertex. 

According to \Cref{lem:findmin3}, 
we have
$
\p[r_i(\vmin) \leq \pi_i /2] \geq \p[r_i(\vmin) >  \pi_i/2]
$. 
The probability that $\finda_i$ is good equals $\p[r_i(\vmin) \leq \pi_i /2]$, which is larger than half.

\end{proof}

{
We directly have the following corollary.
\begin{corollary}
\label{lem:comA}
The number of rounds where $\finda_i$ happens for all $i$ is $O(k\log n)$ in expectation. 
\end{corollary}
}
\begin{proof}
For a fixed $i$, the number of $\finda_i$ is $O(\log n)$ in expectation. Summing up $k$ chains, it is $O(k\log n)$.    
\end{proof}

Then, we bound the times of event $\findb_i$ happens.
{
\begin{lemma}
\label{lem:comBpre}
If for some $i$, $\findb_i$ happens on $b$, $\findb_j$ will not happen on $b' \succeq b$ in later rounds for all $j$.  

\end{lemma}
}
\begin{proof}
	If $\findb_j$ does happen on $b' \succeq b$ in later rounds after $\findb_i$ happens on $b$, according to \Cref{lem:findmin2fo}, $\exists a\in A_{\succ b'}\subseteq A_{\succ b}$ which is still in $\Tilde{A}$, which contradicts to the fact we have deleted $A_{\succ b}$ from $\Tilde{A}$.
\end{proof}

\begin{lemma}
\label{lem:comBpre2}
When $\findb_i$ happens on vertex $b_{ij}$ on chain $i$, $\finda_i$ must have happened to all vertices $a_{ij'}$ for $j'\leq j$.
\end{lemma}
\begin{proof}
First, when $\findb_i$ happens on $b_{ij}$, $a_{ij'}$ for $j'\leq j$ must have been deleted as a result of \Cref{lem:findmin2}.
Second, $\forall a_{ij'},\ j'=1,...,j$ are not deleted by some $B$-side vertices $b$ in line 11, otherwise $\findb_i$ won't happen to $b_{ij} \succeq b$ (\Cref{lem:comBpre}). They can only be deleted by themselves due to $\finda_i$ happening on $a_{ij'}$. Therefore, when $\findb_i$ happens to $b_{ij}$, $\finda_i$ must have happened on vertices $a_{ij'}$ for all $j'\leq j$.
\end{proof}

{
\begin{lemma}
\label{lem:comB}
For each chain $i$, the number of rounds where $\findb_i$ happens is at most the number of rounds where $\finda_i$ happens. 
\end{lemma}
}

\begin{proof}
According to \Cref{lem:comBpre2}, if $\findb_i$ happens to $b_{ij}$, $\finda_i$ must have happened to $a_{ij}$. Since $\findb_i$ happens to a vertex $b_{ij}$ at most once according to \Cref{lem:comBpre}, we can charge $\findb_i$ happens to $b_{ij}$ to $\finda_i$ happens to $a_{ij}$, while promising $a_{ij}$ is charged at most once. 
Therefore the number of rounds where $\findb_i$ happens is at most the number of rounds where $\finda_i$ happens.
\end{proof}

Therefore, the times of $\findb_i$ happens is no larger than the times of $\finda_i$ happens. The query complexity and the time complexity is both bounded in $O(nk\log n)$ in expectation, which concludes the proof of \Cref{lem:bipartitemain}.      \section{Generalized Poset Sorting with Comparable Edges}

In this section, we discuss the GPSC problem and prove \Cref{thm:gpsc}. 

\thmgpsc* 

The algorithmic framework is similar to the augmentation framework in GPS. However, instead of using a correct linear extension for help, we apply an ingredient that can make a rough prediction for every edge's direction, to guide our augmentation. We remark that in GPSC, a linear extension already suffices to reveal the poset in GPSC. First, we introduce how we construct the predictor. 

\subsection{Direction Predictor} 

The direction predictor is a directed graph $\tilde{G}=(V,\tilde{E})$, where each edge in $\tilde{E}$ has a predicted direction. We call an edge in $\tilde{E}$ wrong if its direction is different from $\vec{E}$. We claim that we can use $\tilde{O}(n\sqrt{n})$ queries to construct a predictor such that for every vertex $v$, the number of wrong edges adjacent to $v$ is bounded. 
This subroutine is inspired by the in-degree predictor in \cite{DBLP:conf/focs/HuangKK11}. 
In \cite{DBLP:conf/focs/KuszmaulN21}, the authors propose a predictor with a bounded number of wrong edges. However, these predictors do not directly imply our predictor. Remark that the construction we present is not efficient in the sense of running time. We omit to discuss this computational challenge since we mainly focus on query complexity. But we believe it is possible to make it efficient by the same average rank technique in \cite{DBLP:conf/focs/HuangKK11}.

\begin{lemma}
\label{lem:wrongedge}
With high probability, by using $\tilde{O}(n\sqrt{n})$ queries, we can construct a predictor $\tilde{G}=(V,\tilde{E})$, the number of wrong edges for each vertex is at most $\tilde{O}(\sqrt{n})$.
\end{lemma}
\begin{proof}
    The basic prediction idea is inspired by \cite{DBLP:conf/focs/HuangKK11}. Consider we already know some directions in $E$ and focus on the kinds of linear extensions that are still feasible. In particular, we call a linear extension $\{v_1,v_2,\dots,v_n\}$ feasible for a known directed edge set $\bar{E}$ if we do not have $v_i$ can reach $v_j$ in $\bar{E}$ but $i>j$. Then, how to predict an unknown edge? The two different answers of the edge ($u\prec v$ or $u \succ v$, where we do not have $u \nsim v$ in GPSC) should correspond to different linear extensions that are still feasible. 
    We simply enumerate all the possibilities of linear extensions and predict the direction with more feasible linear extensions. We remark that this task may take exponential time to complete. In \cite{DBLP:conf/focs/HuangKK11}, they design an efficient way to approximately realize this prediction idea by an average rank technique. However, our paper mainly focuses on query complexity, so we omit this computational challenge. 

    Building on this basic prediction idea, we move to a vertex testing subroutine that brings us a good predictor for each vertex. For a specific vertex $v$, we randomly sample $\sqrt{n}$ edges adjacent to the vertex and query them. We either find a wrong prediction or say the vertex passes the test. If we find a wrong edge, we re-predict everything and randomly sample $\sqrt{n}$ edges again. We have that 
    \begin{itemize}
        \item If we pass the test, with high probability, the number of wrong edges adjacent to $v$ is at most $O(\sqrt{n}\log n)$ because its degree is at most $n$. We use $\tilde{E}_v$ to record the current prediction and call it the vertex predictor for $v$. It will not change in later rounds. 
        \item In each unpassed test, at least one wrong edge is queried and the number of feasible linear extensions is decreased by at least a half. There are at most $O(n\log n)$ unpassed tests because we only have $n!$ possible linear extensions at the beginning. At most $O(n \sqrt{n} \log n)$ queries are spent on these unpassed tests. 
\end{itemize}

    The good property from the vertex testing is that we have a good vertex predictor for each $v$. In particular, with high probability, we can use $\beta = \Theta(\sqrt{n}\log n)$ to bound the number of wrong edges in all $\tilde{E}_v$.
    However, these predictors may not be consistent. For example, we may predict $u\rightarrow v$ in $\tilde{E}_u$ but $v \rightarrow u$ in $\tilde{E}_v$. 

    The final task is to construct a global predictor $\tilde{E}$ that is good for each $v$. First, we simply fix $\tilde{E}$ as the final predictor after we test the last vertex and compare it to every $\tilde{E}_v$. If the difference between $\tilde{E}$ and any $\tilde{E}_v$ is larger than $2\beta$, we will sample that vertex to adjust $\tilde{E}$. Because $\tilde{E}_v$ at most have $\beta$ wrong edges, so at least half of these different edges are wrong in $\tilde{E}$. We keep querying one of these different edges randomly, and we can find one wrong edge after $O(\log n)$ queries with high probability. Every time we find a wrong edge, we re-predict $\tilde{E}$ by the basic prediction. As a result, because we can at most query $O(n\log n)$ wrong edges in $\tilde{E}$, we will not enter this case after $O(n\log n)$ queries. On the other hand, if the difference between $\tilde{E}$ and each $\tilde{E}_v$ is already bounded in $2\beta$, we are done because, for each $v$, we have at most $3\beta = \tilde{O}(\sqrt{n})$ wrong edges. $\tilde{E}$ is a good predictor that satisfies the lemma. 
\end{proof}

Notice that the property holds for high probability. In the following section, we assume the property always holds when we use the predictor.

\subsection{Sorting with Predictor}

Next, we solve $\mathcal{P}=\mathcal{P}(\vec{G})$ in an augmentation fashion with the help of the predictor. We keep inserting new vertex $v$ into a maintained vertex subset $A$. After that, we should reveal the real direction of all predicted incoming edges of $u$ in $\tilde{E}$, by performing some queries. Finally, when $A$ becomes $V$, we are done. The most critical idea is to show we can always find a proper new vertex $u$ that we can augment with  small number of queries.

Assume the current vertex set is $A$. The naive idea is to find a minimal vertex in $V \setminus A$ w.r.t. $\mathcal{P}(\tilde{G})$. Ideally, if $\tilde{G}$ is precise, we can prove that the width of $\mathcal{P}(\vec{G}[A])$ is always bounded by $k$. Therefore, when we insert a new vertex, its predicted incoming edges are all in $A$, which can be decomposed into $k$ chains. We can finish this augmentation in $O(k \log n)$ queries by applying a binary search of the new vertex on each chain. However, the challenge is that $\tilde{G}$ may contain wrong edges, which makes the width of $\mathcal{P}(\vec{G}[A])$ no longer bounded by $k$.

Let $\tilde{I}_u$ be the set of predicted incoming vertices (w.r.t. $\tilde{G}$) of a specific vertex $u$. The cost of inserting $u$ into $A$ depends on the minimum number of chains where we can decompose $\tilde{I}_u$ based on the already revealed information. Remark that now we have already known all the real directions of the predicted incoming edges to vertices in $A$, which means we at least already know $\mathcal{P}(\vec{G}[A])$. Next, we prove that there must be a good vertex $u \in V \setminus A$, s.t. we can decompose $\tilde{I}_u$ into $k + \tilde{O}(\sqrt{n})$ chains. The two different terms ($k$ and $\tilde{O}(\sqrt{n})$) in the lower bound come from the chain decomposition for two different subsets of $\tilde{I}_u$. We define them as follows, 
$$
\tilde{I}_u^{\top} = \{v \in V \mid (v\rightarrow u) \in \tilde{E}, (u \rightarrow v) \in \vec{E}\},
\quad \tilde{I}_u^{\bot} = \{v \in V \mid (v\rightarrow u) \in \tilde{E}, (v \rightarrow u) \in \vec{E}\}.
$$

\begin{lemma}
    $\forall u \in V \setminus A$, we can decompose $\tilde{I}_u^{\bot}$ into $\tilde{O}(\sqrt{n})$ chains. 
\end{lemma}
\begin{proof}
    This lemma directly follows from \Cref{lem:wrongedge}. With high probability, the prediction error is bounded for every vertex, so $|\tilde{I}_u^{\bot}| = \tilde{O}(\sqrt{n})$. We can make one chain for each of them. 
\end{proof}
\begin{lemma}
\label{lem:truechaincover}
$\exists u\in V\setminus A$, such that we can decompose $\tilde{I}_u^{\top}$ into $k$ chains w.r.t. to $\mathcal{P}(\vec{G}[A])$.
\end{lemma}
\begin{proof}
At first, we know that $\vec{G}$ can be decomposed into $k$ chains. Refer to \Cref{fig:insertv}, we mark vertices in $A$ in blue and the others in grey. We emphasize the first grey vertex on each chain and call them boundary vertices. Let us consider a minimal grey vertex $u^*$, which should be one of the boundary vertices. Because $u^*$ is minimal (i.e., no other grey vertices can be smaller than $u^*$), all vertices smaller than $u^*$ should be on the left-hand side of the boundary vertices. Therefore, on each chain in the figure, we already have a directed path between every two vertices in $\vec{G}[A]$, which means any subset of the chain should also be a chain w.r.t. $\mathcal{P}(\vec{G}[A])$. Therefore, we can partition $\tilde{I}_u^{\top}$ into at most $k$ chains w.r.t. $\mathcal{P}(\vec{G}[A])$.
\end{proof}

\begin{figure}[H]
    \centering
    \includegraphics[width=.7\textwidth]{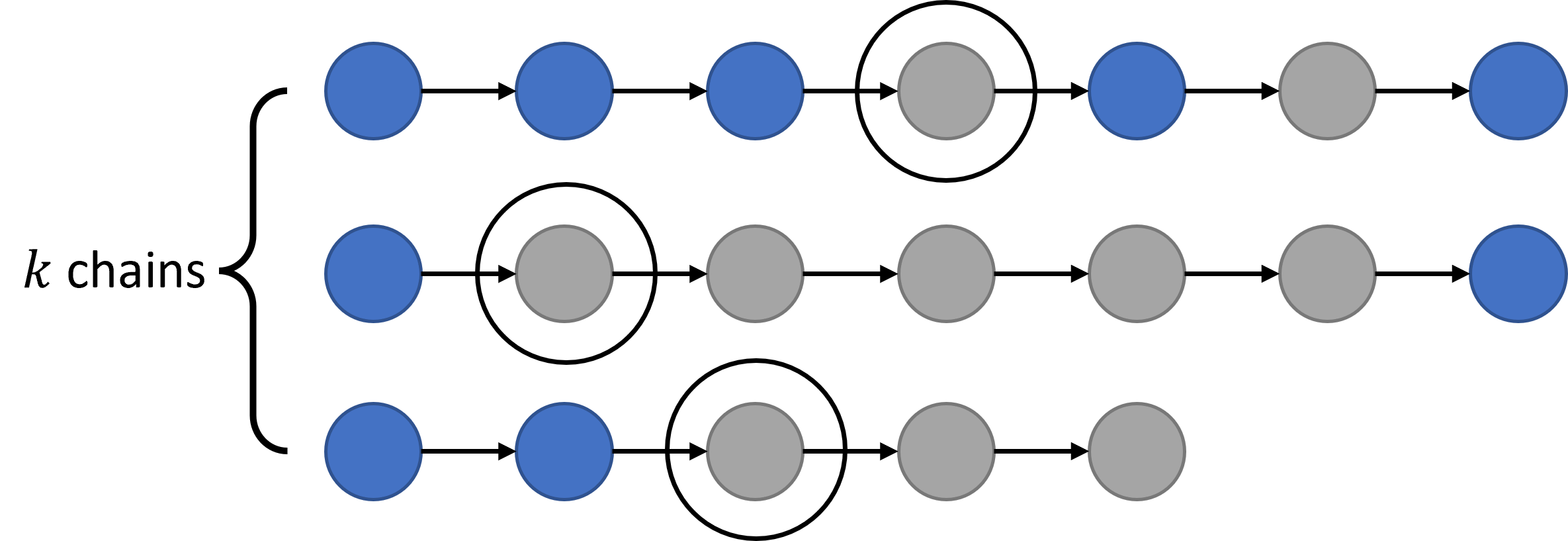}
    \caption{The $k$ chain cover on $\vec{G}$. Blue vertices represent the vertex set $A$.}
    \label{fig:insertv}
\end{figure}

By combining the two lemmas above, we know the existence of a good choice. We can try to decompose $\tilde{I}_u$ for every $u$ and select the best one. Noticed that it can be done efficiently and without any extra queries. Our algorithm is formalized below:

\begin{algorithm}[H]
\caption{Solving $\mathcal{P}$ in GPSC}
\label{alg:gpsc}
\begin{algorithmic}[1]
\Procedure{\sortpartial}{$G$}
\State underlying input: $\mathcal{P}$ and $\vec{G}$, where $\mathcal{P}=\mathcal{P}(\vec{G})$
\State construct $\tilde{G}$ by $\tilde{O}(n\sqrt{n})$ queries
\State $A \gets \emptyset$
\While{$|A| < |V|$}
\State decompose $\tilde{I}_u$ into chains by already known $\mathcal{P}(\vec{G}[A])$ for every $u$
\State select the best $u$ where we can use the minimized number chains to decompose $\tilde{I}_u$
\State apply binary search on edges between all decomposed chains and $u$
\State $A = A \cup  \{u\}$
\EndWhile
\EndProcedure
\end{algorithmic}
\end{algorithm}

Finally, the proof of \Cref{thm:gpsc} follows straightforwardly: with high probability, the predictor is good so that we can finish every augmentation in $\tilde{O}(k+\sqrt{n})$ queries. Thus, our algorithm totally needs $\tilde{O}(nk + n\sqrt{n})$ queries, including the queries for constructing the predictor.

     \section{Weighted Generalized Sorting}

In this section, we move to the weighted setting. We use $\{w_1<w_2<w_3\dots\ w_W\}$ to denote the set of different weights that appeared in the input, where $W$ stands for the number of types of weights. 
As stated in \Cref{thm:weighted}, we prove that we can achieve a competitive ratio of $\tilde{O}(n^{1-1/(2W)})$ (restated below).

\thmweighted*

In particular, 
we use $(v_1,v_2,\dots,v_n)$ to denote the directed Hamiltonian path, and we use $\opt = \sum_{i=2}^n w(v_{i-1},v_{i})$ to denote the cost of the optimal solution. The competitive ratio is defined to be the worst case ratio between the algorithm's query cost to \opt. In particular, to prove this theorem, we show that our query cost is bounded in $\tilde{O}(n^{1-1/(2W)}) \cdot \opt$. In this section, to simplify notations, we use $\exp_n(x)$ to mean $n^x$ later. 

Our weighted algorithm uses the doubling trick. In particular,
we guess \opt from the smallest positive edge weight. Then we run our algorithm called \sortconstant with $\eopt$ (the estimated \opt) as an input. We restrict the query cost we use to $\beta = \tilde{O}\left(\exp_n(1-\frac{1}{2W})\right)\cdot \eopt$ (we will terminate the subroutine when the cost of queries exceeds $\beta$). If \sortconstant fails to sort the input, we double $\eopt$ and run $\sortconstant(\eopt)$ again. We prove that $\sortconstant(\eopt)$ can successfully sort the input with $\beta = \tilde{O}(\exp_n\left(1-\frac{1}{2W})\right)\cdot \eopt$ cost when $\eopt\geq 2 \cdot \opt$. Therefore, \sortconstant can successfully sort the input at the first time $\eopt \geq 2 \cdot \opt$, and the blow-up of the competitive ratio of this doubling process is only a constant factor. Thus, the total query cost we use is controlled in $O(\beta) = \tilde{O}(\exp_n\left(1-\frac{1}{2W})\right)\cdot \eopt.$

In the algorithm \sortconstant, we will fix a threshold $\tau$. Then, we use \sortpartial on the edges at most $w_{\tau}$ (denoted by $E^{\leq w_{\tau}}$) to get poset $\mathcal{P}_\tau$. Let $\vec{E}^{\leq w_{\tau}}$ be the directed counterpart of the small cost edges in the underlying graph. We remark that $\mathcal{P}_\tau$ is an induced poset by the underlying subgraph $\vec{G}_\tau = (V, \vec{E}^{\leq w_{\tau}})$. We use $k_\tau$ to denote its width, and we know we can decompose vertices into $k_\tau$ chains w.r.t. $\mathcal{P}_\tau$. Next, we use a subroutine called \sortchain to sort these chains and get the final order. The cost of \sortchain depends on $k_{\tau}$. The algorithm is formalized in \Cref{alg:sortconstant}, and the subroutine \sortchain is introduced in the next subsection. 

\begin{algorithm}[H]  
\caption{Sorting Algorithm with An Estimated \opt}
\label{alg:sortconstant}
    \begin{algorithmic}[1]
        \Procedure{$\sortconstant(G,\eopt)$}{}
            \State use $\findj(\eopt)$ to calculate the threshold $\tau$.
            \State $G_{\tau} \gets (V, E^{\leq w_{\tau}})$, where $E^{\leq w_{\tau}}$ is the set of edges with cost at most $w_{\tau}$.
            \State $\mathcal{P_\tau} \gets \sortpartial(G_{\tau})$.
            \State decompose $V$ into $k_\tau$ chains w.r.t. $\mathcal{P_\tau}$. 
            \State $\sortchain(G,\mathcal{C}$).
        \EndProcedure
    \end{algorithmic}
\end{algorithm}

\subsection{Sorting $d$ Chains}

In this section, we present the subroutine \sortchain that aims to sort all vertices under the input of $d$ sorted chains $\mathcal{C}=\{C_1,C_2,...,C_d\}$. The cost is bounded by $O(d\log n) \cdot \opt$. 

In general, the subroutine \sortchain determines vertices one by one in the ascending order of their final rank. Recall that we use $v_1 \prec v_2 \prec \ldots \prec v_n$ to mean the correct order. In particular, the algorithm detects who is $v_1$ first, then $v_2,v_3,\ldots,v_n$. 

To control the query cost, we query edges from small cost to large cost. We maintain a \emph{weight level} of each vertex $u$ as $\ell(u)$ to show how large edges adjacent to $u$ we have already queried. We equip with a subroutine called \textsc{Probe($u$)}, where we push $u$ to the next weight level and determine the direction of all new-level edges adjacent to $u$ in $O(d\log n)$ comparisons. Notice that when a vertex $u$ reaches weight level $\ell(u)$, it means we have called $\ell(u)$ times of \textsc{Probe($u$)} and know the direction of all edges adjacent to $u$ with weight at most $2^{\ell(u)}$. The subroutine is formally presented in \Cref{alg:probe}.

\begin{algorithm}[H]
	\caption{Raise $u$'s Weight Level}
	\label{alg:probe}
    \begin{algorithmic}[1]
        \Procedure{Probe}{$u$}
            \State $\ell(u) \gets \ell(u)+1$
            \State $w \gets 2^{\ell(u)}$
            \For{each chain $C_i$} 
\State apply the binary search for edges with costs at most $w$ between $u$ and $C_i$
\EndFor
            \State update $I_u$ \Comment{$I_u$ maintains all incoming vertices known to our algorithm}
        \EndProcedure
\end{algorithmic}
\end{algorithm} 

The main part of the algorithm is how we determine $v_j$ with $v_1,v_2,\cdots, v_{j-1}$. Assume $A=\{v_1,v_2,\ldots,v_{j-1}\}$, a vertex $u$ is the rank-$j$ vertex $v_j$ if and only if the all incoming vertices to $u$ is in $A$. 
However, when running the algorithm, we may only know the direction of a part of the adjacent edges of $u$ (because of the weight level), so we only know a current incoming vertex set of $u$, called $I_u$. There may be more than one vertex $u$ such that $I_u$ is a subset of $A$. Intuitively, they are all candidates of the real $v_j$. Then, we keep calling \textsc{Probe($u$)} to the candidate with the lowest weight level to shrink the size of the candidate set. Finally, if there is only one candidate left, it must be the real $v_j$. The process can be done in finite steps because there must leave only one candidate after all vertices are pushed to the top weight level. The subroutine \sortchain is formalized in Algorithm \ref{alg-total}.

\begin{algorithm}[H]   
\caption{Sort $d$ Chains}
\label{alg-total}
    \begin{algorithmic}[1]
        \Procedure{SortChain}{$\mathcal{C}=\{C_1,C_2,\ldots,C_d\}$}
            \State $A \gets \emptyset$
            \State $w' \gets$ the minimum positive weight
            \State $\ell(u) \gets \lceil \log {w'}\rceil -1$ for all $u$
\For{$j=1$ to $n$}
            	\State $\Delta = \{u \mid I_u\text{ is a subset of }$A$\}$
                \While{$|\Delta|>1$}
                    \State $u$ = $\arg\min_{\{ u \in \Delta\}} \ell(u)$
                    \State \textsc{Probe($u$)}
                    \State update $\Delta$
                \EndWhile
                \State fix $v_k$ to be the only candidate in $\Delta$
            \EndFor
        \EndProcedure
    \end{algorithmic}
\end{algorithm}

Finally, we prove the query cost of \sortchain is upper bounded by $O(d\log n)\cdot\opt$. 
In our analysis, we charge all query costs to vertices and prove the total budget of each $v_j$ required is $O(w(v_{j-1},v_j)\cdot d\log n)$. It concludes the desired cost bound because $\opt=\sum^n_{i=2}w(v_{i-1},v_i)$. 

There are two kinds of charging. 
\begin{enumerate}
    \item In the $i$-th round where $i < j$, the cost of \textsc{Probe}$(v_j)$ is charged to $v_j$. 
    \item In the $j$-th round, the cost of \textsc{Probe}$(v_j)$ is charged to the last eliminated candidate $u$ in this round. 
\end{enumerate}

Therefore, for a specific vertex $v_j$, it will be only charged in rounds before $j$, and in two different ways: 1) by \textsc{Probe}$(v_j)$ itself, and 2) by other \textsc{Probe}$(v_{j'})$ at $j'$-th round. We discuss the two cases in \Cref{lem:sort-subset-bound-1} and \Cref{lem:sort-subset-bound-2}. Before that, we state the query cost of \textsc{Probe}$(u)$.

\begin{lemma}
\label{lem:probeu_cost}
We have the following two properties of \textsc{Probe}$(u)$.
\begin{enumerate}
    \item \textsc{Probe}$(u)$ costs at most $O(2^{\ell(u)}\cdot d \log n)$, where $\ell(u)$ is the weight level of $u$ after \textsc{Probe}$(u)$.
    \item When $u$'s level grows to $\ell(u)$, the sum of cost of \textsc{Probe}$(u)$ occurs is  $O(2^{\ell(u)}\cdot d \log n)$
\end{enumerate}

\end{lemma}
\begin{proof}
At first, the number of edges we query is at most $O(\log n)$ on each chain, so we have $O(d\log n)$ in total. Then, the lemma holds because we only query edges with weight at most $2^{\ell(u)}$. 
Then, we use the fact that each \textsc{Probe}$(u)$ will push $u$ to the next weight level. Therefore, when $u$'s level is $\ell(u)$, the total cost is 
$$
O\left(2^{\log w'} d\log n + 2^{\log w'+1} d\log n +\ldots + 2^{\ell(u)} d\log n\right) 
= O\left(2^{\ell(u)}\cdot d\log n\right)
$$
\end{proof}

\begin{lemma}
\label{lem:sort-subset-bound-1}
For every $v_j$, the summation of the cost charged to $v_j$ by \textsc{Probe}$(v_j)$ can be bounded by $O(w(v_{j-1},v_j)\cdot d\log n)$.
\end{lemma}
\begin{proof}
The first fact is that $v_j$ won't be a candidate in $\Delta$ if the direction of $(v_{j-1},v_j)$ is already known in the rounds before $j$. 
Therefore, the charging stops when $v_j$'s weight level grows to $\ell$ above $(v_{j-1},v_j)$, i.e., $w(v_{j-1},v_{j}) \leq 2^{\ell} \leq 2w(v_{j-1},v_{j})$. By \Cref{lem:probeu_cost}, the sum of costs of \textsc{Probe}$(v_j)$ before $j$-th round is totally 
\begin{align*}
    O\left(2^{\ell} \cdot d\log n\right)
    = O\left(w(v_{j-1},v_j)\cdot d\log n\right)
\end{align*}
\end{proof}

\begin{lemma}
\label{lem:sort-subset-bound-2}
For every $v_j$, the summation of the charged cost by \textsc{Probe}$(u)$ from vertices $u$ other than $v_j$ is bounded by $O(w(v_{j-1},v_j)\cdot d \log n)$.
\end{lemma}
\begin{proof}
    We have that the cost of \textsc{Probe}$(u)$ is charged to $v_j$ only at the $i$-th round where $i$ is $u$'s rank, and $v_j$ is the last eliminated vertex. Define $\ell_i(v_j),\ell_i(u)$ to be the weight level of $v_j$ and $u$ when round $i$ ends. We have that $\ell_i(u) \leq \ell_i(v_j)$ because we always select to push the lowest weight level. Even assume all the cost of \textsc{Probe}$(u)$ is charged to $v_j$, it can still be bounded by 
    $$
    O\left(2^{\ell_i(u)} \cdot d\log n\right)
    = O\left(2^{\ell_i(v_j)} \cdot d\log n\right), 
    $$
    which is at the same order as the cost of \textsc{Probe}$(v_j)$ at level $\ell_i(v_j)$. To be more precise, we say these costs are charged to $v_j$ at level $\ell_i(v_j)$. Another fact is that whenever this charging happens, 
    where $v_j$ is a candidate and then eliminated, the weight level should be pushed at least once. Therefore, each weight level of $v_j$ should be charged by at most one other vertex. As a result, letting $\ell$ be the last weight level of $v_j$ just before $j$-th round, even assuming every weight level of $v_j$ before $j$-th round is charged, the total cost is bounded by
    $$
    O\left(2^{\ell}\cdot d\log n\right)
    = O\left(w(v_{j-1},v_j)\cdot d\log n\right).
    $$
    Again it is because $2^{\ell} \leq 2w(v_{j-1},v_{j})$ since $v_j$'s level can not been push above $w(v_{j-1},v_{j})$ before $j$-th round.     
\end{proof}

\subsection{Find A \emph{good} Threshold}
Finally, we prove that we can find a good threshold $\tau$, such that the following lemma can be proved. 

\begin{lemma}
    \sortconstant (\Cref{alg:sortconstant}) can output the total order with query cost at most $\tilde{O}\left(\exp_n(1-\frac{1}{2\W})\right)\cdot \eopt$ when $\eopt \geq 2\cdot \opt$.
\end{lemma}

First, let us formally state the cost of \sortconstant by simply combing the cost of \sortpartial and \sortchain.
\begin{equation}
    \label{eqn:sortconstantcost}
    \tilde{O}\left((nk_{\tau} + n^{1.5})\cdot w_{\tau} + k_{\tau}\cdot \opt\right).
\end{equation}

We first ruin a trivial case when $w_\W \leq \exp_n(-\frac{1}{2\W}-0.5)\cdot \eopt$, by choosing $\tau = \W$

\begin{lemma}
\label{lem:sortconstantsimple}
    If $w_\W \leq \exp_n(-\frac{1}{2\W}-0.5)\cdot \eopt$, the query cost is $\tilde{O}\left(\exp_n(1-\frac{1}{2\W})\right)\cdot \eopt$ if we choose $\tau=\W$. 
\end{lemma}
\begin{proof}
    In this case, we have $\tau = \W$ and $k_{\tau}=1$. Putting it into \Cref{eqn:sortconstantcost}, the cost is
    $$
    \tilde{O}\left(\exp_n(1.5 -\frac{1}{2\W}-0.5)\right) \cdot \eopt = \tilde{O}\left(\exp_n(1-\frac{1}{2\W})\right)\cdot \eopt
    $$ 
\end{proof}

Then, we move toward the case $w_\W > \exp_n(-\frac{1}{2\W}-0.5)\cdot \eopt$. We need to select $\tau<m$, with three constraints listed below.
\begin{align}
    \label{eqn:sortconstantbound1-1}
    &\frac{w_{\tau}}{w_{\tau+1}} \leq \exp_n(-\frac{1}{2\W}).\\
    \label{eqn:sortconstantbound1-2}
    &w_{\tau} \leq \exp_n(-0.5-\frac{1}{2\W}) \cdot \eopt.\\
    \label{eqn:sortconstantbound2}
    &k_{\tau} \leq \exp_n(1-\frac{1}{2\W})
\end{align}

Then, we have two tasks. 1) Prove the three conditions are sufficient. 2) Show we can find a $\tau$ satisfying the three conditions. We prove the first task in the following lemma. 

\begin{lemma}
    \label{lem:weightcondition}
    If we choose $\tau$ that satisfies \Cref{eqn:sortconstantbound1-1}, \Cref{eqn:sortconstantbound1-2}, and \Cref{eqn:sortconstantbound2}, the query cost is $\tilde{O}\left(\exp_n(1-\frac{1}{2\W})\right)\cdot \eopt$. 
\end{lemma}
\begin{proof}
    Recall the cost in \Cref{eqn:sortconstantcost}. We prove 
    $$
    \frac{(nk_{\tau} + n^{1.5})\cdot w_{\tau}}{\eopt} + k_{\tau} = O(\exp_n(1- \frac{1}{2\W}))
    $$ 
    by partitioning the LHS into three different terms.
    
    At first, because $k_{\tau}$ is the width of $\mathcal{P}_{\tau}$. We can lower bound $\eopt$ by 
    $$
        \eopt \geq \opt \geq w_{\tau+1} \cdot (k_{\tau}-1).
    $$
    Therefore, we bound the first term. 
    $$
        \frac{n(k_{\tau}-1)w_{\tau}}{\eopt} \leq \frac{n(k_{\tau}-1)w_{\tau}}{w_{\tau+1} \cdot (k_{\tau}-1)} \leq \exp_n(1 - \frac{1}{2\W}).
    $$
    Remark that the final inequality holds by the constraint in \Cref{eqn:sortconstantbound1-1}. Then, 
    We bound the second term by the constraint in \Cref{eqn:sortconstantbound1-2}.
    $$
    \frac{(n + n^{1.5}) \cdot w_{\tau}}{\eopt} \leq 2\exp_n(1- \frac{1}{2\W}).
    $$
    Finally, we bound the third term directly by \Cref{eqn:sortconstantbound2}.
    $$
    k_{\tau} \leq \exp_n(1-\frac{1}{2\W}).
    $$
    Combining the three terms together, we conclude the lemma. 
\end{proof}

Finally, it remains to complete the second task, i.e., construct an algorithm called \findj to find a good $\tau$. We start from the first weight that exceeds $\exp_n(-0.5-\frac{1}{2\W}) \cdot \eopt$, called $\hat{\tau}$. Then, we enumerate $\tau$ from $\hat{\tau}-1$ downto $1$. We stop at the first time when $\frac{w_{\tau}}{w_{\tau+1}} \leq \exp_n(-\frac{1}{2\W})$ and select it. We formally present the subroutine in \Cref{alg:findj}, combining with the simple case in \Cref{lem:sortconstantsimple}

\begin{algorithm}[H]   
    \caption{Finding A Good Threshold}
    \label{alg:findj}
    \begin{algorithmic}[1]
        \Function{\findj}{$\eopt$}
        \If{$w_\W \leq \exp_n(-\frac{1}{2\W}-0.5) \cdot \eopt$}
            \State \textbf{return} $\tau = \W$.
        \Else 
\State $\hat{\tau}$ = $\min\{j \mid s.t.~ w_j>\exp_n(-\frac{1}{2\W}-0.5)\cdot \eopt\}$
\For{$\tau$ from $\hat{\tau}-1$ down to $1$} 
                \If{$\frac{w_{\tau}}{w_{\tau+1}}\leq \exp_n(-\frac{1}{2\W})$}
\State \textbf{return} $\tau$
                \EndIf
\EndFor
        \EndIf
        \EndFunction
    \end{algorithmic}
\end{algorithm}  

 We prove that \findj gives us a good $\tau$ that satisfies all three constraints. 
\begin{lemma}
    If $w_\W> \exp_n(-\frac{1}{2\W}-0.5)\cdot \eopt$, \Cref{alg:findj} must output $\tau \geq 1$ that satisfies \Cref{eqn:sortconstantbound1-1}, \Cref{eqn:sortconstantbound1-2}, and \Cref{eqn:sortconstantbound2}. 
\end{lemma}
\begin{proof}
We first prove that \Cref{alg:findj} can always find a feasible $\tau$ between $1$ and $\hat{\tau}-1$. Assume the contradiction, we know that $\forall 1\leq j\leq \hat{\tau}-1$, $\frac{w_{j}}{w_{j+1}} > \exp_n(-\frac{1}{2W})$. Thus, we have 
$$
    w_1 > \exp_n\left(-\frac{1}{2\W}-0.5-\frac{\W-1}{2\W}\right)\cdot \eopt > \exp_n(-1) \cdot \eopt.
$$
If we have zero weight edges, where $w_1=0$, there is a direct contradiction. If we do not have zero weight edges, it contradicts the condition of $\eopt \geq 2 \cdot \opt \geq \frac{n}{n-1}\cdot \opt \geq n w_1$ when $n>1$. We remark that we also have $\hat{\tau}>1$ because of the proof above. Thus far, we already know we can output $\tau$ that satisfies \Cref{eqn:sortconstantbound1-1} and \Cref{eqn:sortconstantbound1-2}.

Next, we prove that $k_{\tau}$ is small with our selection of $\tau$. Because $\tau$ is the first one satisfies \Cref{eqn:sortconstantbound1-1}, so we have $\frac{w_{j}}{w_{j+1}} > \exp_n(-\frac{1}{2W})$ for all $\tau + 1 \leq j \leq \hat{\tau}$. Therefore
\begin{align*}
    \frac{w_{\tau+1}}{\eopt} 
    & \geq \frac{w_{\hat{\tau}}}{\eopt} \cdot \exp_n\left(-\frac{1}{2\W}\cdot(\hat{\tau}-1-\tau)\right) \\
    & \geq \exp_n\left(-\frac{1}{2\W}-0.5-\frac{1}{2\W}\cdot (\hat{\tau}-1-\tau)\right)  \\
    & \geq \exp_n\left(-\frac{1}{2\W}-0.5-\frac{1}{2\W}\cdot (\W-2)\right) \\
    & \geq \exp_n(-1 + \frac{1}{2\W})
\end{align*}
By the bound of $\eopt > \opt \geq w_{\tau+1}\cdot k_\tau$. We have $k_\tau \leq \exp_n(1-\frac{1}{2\W})$, that satisfies \Cref{eqn:sortconstantbound2}, which concludes the lemma. 
\end{proof}

Finally, we notice that our method can lead to better bounds if the weights have large gaps. 
We capture this in the following corollary.
\begin{corollary}
    \label{cor:weighted_fast_inc}
If $\forall 1\leq i \leq W$, $\frac{w_i}{w_{i+1}} \leq \exp_n(-3/4)$. We can find a good $\tau$ with query cost $\tilde{O}(\exp_n(3/4))\cdot \eopt$.
\end{corollary}
\begin{proof}
    First, if $w_W \leq \exp_n(-3/4) \cdot \eopt$, then the query cost is $\exp_n(3/4)$, by fixing $\tau=\W$.
    In the remaining case when $w_W$ is large, 
    We can prove stronger versions of  \Cref{eqn:sortconstantbound1-1}, \Cref{eqn:sortconstantbound1-2}, and \Cref{eqn:sortconstantbound2}.
    By the condition of the lemma, we naturally have the stronger version of \Cref{eqn:sortconstantbound1-1}
    $$
    \frac{w_\tau}{w_{\tau+1}} \leq \exp_n(-3/4).
    $$
    We also find $\hat{\tau}$ to be the minmizied index such that $w_{\hat{\tau}}>\exp_n(-3/4)\cdot \eopt$, where we can can prove $\hat{\tau}>1$ because $w_1 \leq n^{-1} \cdot \eopt$. Therefore, we can fix $\tau = \hat{\tau}-1\geq 1$ such that
    
    $$
    k_{\hat{\tau}-1} \leq \frac{\eopt} {w_{\hat{\tau}}} \leq \exp_n(3/4), 
    \quad w_{\hat{\tau}-1} \leq \exp_n(-3/4)\cdot \eopt. 
    $$ 
    The two inequality above is exactly the stronger versions of \Cref{eqn:sortconstantbound1-2} and \Cref{eqn:sortconstantbound2}. Finally, by these stronger conditions, we directly prove the cost is bounded in $\tilde{O}(\exp_n(3/4))\cdot \eopt$ by the same calculation in \Cref{lem:weightcondition} (We can view it as the same case as $W=2$).
\end{proof}

    \bibliography{ref.bib}

\newcommand{\etalchar}[1]{$^{#1}$}
\begin{thebibliography}{DKM{\etalchar{+}}11}

\bibitem[ABF{\etalchar{+}}94]{DBLP:conf/soda/AlonBFKNO94}
Noga Alon, Manuel Blum, Amos Fiat, Sampath Kannan, Moni Naor, and Rafail
  Ostrovsky.
\newblock Matching nuts and bolts.
\newblock In {\em {SODA}}, pages 690--696. {ACM/SIAM}, 1994.

\bibitem[ABF96]{DBLP:journals/ipl/AlonBF96}
Noga Alon, Phillip~G. Bradford, and Rudolf Fleischer.
\newblock Matching nuts and bolts faster.
\newblock {\em Inf. Process. Lett.}, 59(3):123--127, 1996.

\bibitem[AKM08]{DBLP:conf/latin/AngelovKM08}
Stanislav Angelov, Keshav Kunal, and Andrew McGregor.
\newblock Sorting and selection with random costs.
\newblock In {\em {LATIN}}, volume 4957 of {\em Lecture Notes in Computer
  Science}, pages 48--59. Springer, 2008.

\bibitem[BJR17]{DBLP:conf/caldam/BiswasJ017}
Arindam Biswas, Varunkumar Jayapaul, and Venkatesh Raman.
\newblock Improved bounds for poset sorting in the forbidden-comparison regime.
\newblock In {\em {CALDAM}}, volume 10156 of {\em Lecture Notes in Computer
  Science}, pages 50--59. Springer, 2017.

\bibitem[BM08]{DBLP:conf/soda/BravermanM08}
Mark Braverman and Elchanan Mossel.
\newblock Noisy sorting without resampling.
\newblock In {\em {SODA}}, pages 268--276. {SIAM}, 2008.

\bibitem[BR15]{DBLP:journals/corr/BanerjeeR15a}
Indranil Banerjee and Dana Richards.
\newblock Sorting under $1$-$\infty$ cost model.
\newblock {\em CoRR}, abs/1508.03698, 2015.

\bibitem[BR16]{DBLP:conf/swat/BanerjeeR16}
Indranil Banerjee and Dana~S. Richards.
\newblock Sorting under forbidden comparisons.
\newblock In {\em {SWAT}}, volume~53 of {\em LIPIcs}, pages 22:1--22:13.
  Schloss Dagstuhl - Leibniz-Zentrum f{\"{u}}r Informatik, 2016.

\bibitem[Bra95]{bradford1995matching}
Phillip~G Bradford.
\newblock Matching nuts and bolts optimally.
\newblock Technical report, 1995.

\bibitem[CFG{\etalchar{+}}02]{DBLP:journals/jcss/CharikarFGKRS02}
Moses Charikar, Ronald Fagin, Venkatesan Guruswami, Jon~M. Kleinberg, Prabhakar
  Raghavan, and Amit Sahai.
\newblock Query strategies for priced information.
\newblock {\em J. Comput. Syst. Sci.}, 64(4):785--819, 2002.

\bibitem[DKM{\etalchar{+}}11]{DBLP:journals/siamcomp/DaskalakisKMRV11}
Constantinos Daskalakis, Richard~M. Karp, Elchanan Mossel, Samantha~J.
  Riesenfeld, and Elad Verbin.
\newblock Sorting and selection in posets.
\newblock {\em {SIAM} J. Comput.}, 40(3):597--622, 2011.

\bibitem[FRPU94]{DBLP:journals/siamcomp/FeigeRPU94}
Uriel Feige, Prabhakar Raghavan, David Peleg, and Eli Upfal.
\newblock Computing with noisy information.
\newblock {\em {SIAM} J. Comput.}, 23(5):1001--1018, 1994.

\bibitem[FT88]{DBLP:journals/siamcomp/FaigleT88}
Ulrich Faigle and Gy{\"{o}}rgy Tur{\'{a}}n.
\newblock Sorting and recognition problems for ordered sets.
\newblock {\em {SIAM} J. Comput.}, 17(1):100--113, 1988.

\bibitem[GJ22]{DBLP:journals/corr/abs-2211-04601}
Mayank Goswami and Riko Jacob.
\newblock Universal sorting: Finding a {DAG} using priced comparisons.
\newblock {\em CoRR}, abs/2211.04601, 2022.

\bibitem[GK01]{DBLP:conf/focs/GuptaK01}
Anupam Gupta and Amit Kumar.
\newblock Sorting and selection with structured costs.
\newblock In {\em {FOCS}}, pages 416--425. {IEEE} Computer Society, 2001.

\bibitem[GK05]{DBLP:conf/approx/GuptaK05}
Anupam Gupta and Amit Kumar.
\newblock Where's the winner? max-finding and sorting with metric costs.
\newblock In {\em {APPROX-RANDOM}}, volume 3624 of {\em Lecture Notes in
  Computer Science}, pages 74--85. Springer, 2005.

\bibitem[GX23]{noisy_sorting}
Yuzhou Gu and Yinzhan Xu.
\newblock Optimal bounds for noisy sorting.
\newblock In {\em {STOC}}. {ACM}, 2023.
\newblock To appear.

\bibitem[HKK11]{DBLP:conf/focs/HuangKK11}
Zhiyi Huang, Sampath Kannan, and Sanjeev Khanna.
\newblock Algorithms for the generalized sorting problem.
\newblock In {\em {FOCS}}, pages 738--747. {IEEE} Computer Society, 2011.

\bibitem[KK03]{DBLP:conf/soda/KannanK03}
Sampath Kannan and Sanjeev Khanna.
\newblock Selection with monotone comparison cost.
\newblock In {\em {SODA}}, pages 10--17. {ACM/SIAM}, 2003.

\bibitem[KMS98]{DBLP:journals/siamdm/KomlosMS98}
J{\'{a}}nos Koml{\'{o}}s, Yuan Ma, and Endre Szemer{\'{e}}di.
\newblock Matching nuts and bolts in o(n log n) time.
\newblock {\em {SIAM} J. Discret. Math.}, 11(3):347--372, 1998.

\bibitem[KN21]{DBLP:conf/focs/KuszmaulN21}
William Kuszmaul and Shyam Narayanan.
\newblock Stochastic and worst-case generalized sorting revisited.
\newblock In {\em {FOCS}}, pages 1056--1067. {IEEE}, 2021.

\bibitem[LRSZ21]{DBLP:conf/sosa/LuRSZ21}
Pinyan Lu, Xuandi Ren, Enze Sun, and Yubo Zhang.
\newblock Generalized sorting with predictions.
\newblock In {\em {SOSA}}, pages 111--117. {SIAM}, 2021.

\bibitem[RY22]{DBLP:journals/corr/abs-2205-15912}
Jishnu Roychoudhury and Jatin Yadav.
\newblock Efficient algorithms for sorting in trees.
\newblock {\em CoRR}, abs/2205.15912, 2022.

\end{thebibliography}
    \bibliographystyle{alphaurl}

    \appendix

    \newpage

    \section{Proof of \Cref{lem:chernoff-non-independent}}
\label{sec:proof_lem:chernoff-non-independent}

\ChernoffVariant*

\begin{proof}
Let $Z_1,Z_2,\ldots ,Z_N$ be i.i.d. Bernoulli random variables with $\Pr[Z_i=1]=\mu$ for every $1\leq i\leq N$. Let $L = (1-\delta)N\mu$. 
For every $1\leq i\leq N$, we have
$$\Pr[Y_1+Y_2+\ldots +Y_i+Z_{i+1}+\ldots Z_N < L] \leq \Pr[Y_1+Y_2+\ldots +Y_{i-1}+Z_i+\ldots Z_N < L]$$

Hence we have $\Pr[Y_1+Y_2\ldots +Y_N< L]\leq \Pr[Z_1+Z_2\ldots +Z_N< L]$. 
By Chernoff bound, $\Pr[Z_1+Z_2\ldots +Z_N<L] \leq \exp(-N\delta^2\mu/2)$. This concludes our proof. 
\end{proof}

\section{Proof of \Cref{lem:CAnc}}
\label{sec:proof_lem:CAnc}

\CAncUB*

\begin{proof}

Recall the recursive definition of ternary search tree. By \Cref{lem:LE-recur-half}, with probability at least $1/2$, the pivot vertex $p$ partitions $X(p)$ into (relatively) even sets. 

\begin{lemma}\label{lem:LE-recur-half}
    For every $X\subseteq V$, $\Pr_{p\sim X}\left[|X_{\prec p}| \leq \frac{3}{4}|X|, |X_{\succ p}| \leq \frac{3}{4}|X|\right] \geq 1/2$. 
\end{lemma}

\begin{proof}
    Let $x_1,x_2,\ldots ,x_\ell$ be a linear extension of $X$. When $x_i$ is chosen as $p$, we have $X_{\prec x_i} \subseteq \{ x_1,x_2,\ldots ,x_{i-1}\}$, $X_{\succ x_i}\subseteq \{x_{i+1},x_{i+2},\ldots ,x_\ell\}$. Hence we have
    $$\Pr_{p\sim X}\left[|X_{\prec p}| \leq \frac{3}{4}|X|, |X_{\succ p}| \leq \frac{3}{4}|X|\right] \geq \Pr_{i\sim [\ell]}\left[\frac{1}{4}|X| \leq i \leq \frac{3}{4}|X|\right] \geq 1/2.$$
    This concludes the proof of \Cref{lem:LE-recur-half}. 
\end{proof}

    Fix some $v$. Let $\Troot \to p_1\to p_2\to \ldots \to p_\ell \to v$ be the path from $\Troot$ to $v$ in $T$. Let $p_{i_1},p_{i_2},\ldots ,p_{i_t}$ ($i_1<i_2<\ldots <i_t,t=|\CAnc(v)|$) be the comparable ancestors of $v$. 
    Let $\ell^*=60\log n$ be the upper bound for $\ell$. In the following discussion, we will show that $\ell \leq \ell^*$ holds with high probability. 
    Define $\ell^*$ indicator variables $Y_1,Y_2,\ldots ,Y_{\ell^*}$ as below. 
    $$Y_j=\begin{cases}
        I(|X(p_{i_j})_\prec| \leq \frac{3}{4}|X(p_{i_j})|\text{ and } |X(p_{i_j})_\succ| \leq \frac{3}{4}|X(p_{i_j})|) & j \leq \ell \\
        1 & j > \ell
    \end{cases}$$
    
    Recall that for every $2\leq j\leq t$, $X(p_{i_j})$ can be uniquely determined by $p_1,p_2,\ldots ,p_{i_j-1}$ according to the recursive definition of ternary search tree. 
    By \Cref{lem:LE-recur-half}, for every $Y_1,Y_2,\ldots ,Y_{i-2}$, $\Pr[Y_{i-1}=1\mid Y_1,Y_2,\ldots ,Y_{i-2}] \geq 1/2$. 
When $Y_{i-1}=1$ holds, we have $|X(p_{i_j})| \leq \frac{3}{4}|X(p_{i_{j-1}})|$ with probability 1. 
    
    Let $\hat{\mu} = 3\log n \geq \log_{4/3} n$, $Y_1+Y_2+\ldots +Y_{\ell^*} \geq \hat{\mu}$ implies $\ell \leq \ell^*$ with probability 1. 

    By \Cref{lem:chernoff-non-independent}, $\Pr[\ell \leq \ell^*] \geq 1 - \Pr[Y_1+Y_2+\ldots +Y_{\ell^*} < \hat{\mu}] \geq 1-n^{-10}$.
    This finishes the proof of \Cref{lem:CAnc}.
\end{proof}

\section{Proof of \Cref{lem:sum-max}}
\label{sec:proof_lem:sum-max}

\SumMax*

\begin{proof}
    Suppose $A=\{a_1,a_2,\ldots ,a_n\}$ ($a_1<a_2<\ldots <a_n$). A uniform permutation of $A$ can be generate in the following way. We iteratively determine $p$ in the order $p_n,p_{n-1},\ldots ,p_1$. 
    For every $1\leq i\leq n$, $p_i$ is randomly distributed among $A\setminus \{p_{i+1},p_{i+2},\ldots n\}$. 

    Define $n$ indicator variables $X_1,X_2,\ldots ,X_n$, where 
    $$X_i = I\left(p_i > \max_{1\leq i'<i} p_{i'}\right) = I(p_i = \max (A\setminus \{p_{i+1},p_{i+2},\ldots n\})).$$

    From our generation method, we know that $X_1,X_2,\ldots ,X_n$ are mutually independent and $\Pr[X_i=1] = 1/i$ for every $1\leq i\leq n$. Let $\mu = \sum_{i=1}^n E[X_i]$, we have $\ln n<\mu < \ln n + 1$. 
    By Chernoff bound, 
    $$\Pr\left[\sum_{i=1}^n X_i > 3\ln(1/\varepsilon) \right] \leq \varepsilon.$$
This finished the proof of \Cref{lem:sum-max}.
\end{proof}

\section{Proof of \Cref{lem:count-adj}}\label{sec:proof-count-adj}

\CountAdj*

\begin{proof}
    Let $s = |S|$. 
    Let $R_S = \{ (u,v) \mid (u,v)\notin \uEbase, u\in S \}$ be a set of unordered vertex pairs.
    For every $(u,v)\in R_S$, we have $\Pr[E(u,v)] = p$. 

    By Chernoff bound, $\Pr[\sum_{(u,v)\in R_S}\mathbb{I}((u,v) \in E) > \max\{4np|S|, 8\log N\} ] \leq N^{-6}$. With probability of at least $1-N^{-6}$, we have 
    $$\sum_{v\in S}|\gAdj_G(v)| = \sum_{v\in S}|\gAdj_{\uGbase}(v)| + \sum_{(u,v)\in R_S}\mathbb{I}((u,v) \in E) \leq ks + \max\{3nps, 6\log N\}.$$

    Notice that when $np < k$, $G$ only contains $O(nk)$ edges in expectation and we can simply query all edges. Hence we assume $np > k$, which implies $ks + \max\{3nps, 6\log N\} \leq \max\{4nps, 8\log N\}$. 

\end{proof} 
\end{document}